\documentclass[11pt]{article}
\usepackage{fullpage}

\usepackage{amssymb,amsmath}
\usepackage{graphicx, epsfig}
\usepackage{tikz}
\usetikzlibrary{positioning}
\usepackage[linesnumbered,lined,boxed,commentsnumbered]{algorithm2e}
\usepackage{wrapfig}

\newcommand{\no}[1]{}

\newcommand{\eps}{\varepsilon}

\newcommand{\point}{\mathtt{point}}
\newcommand{\noderange}{\mathtt{noderange}}

\newcommand{\Patrascu}{P\v{a}tra\c{s}cu}

\newtheorem{lemma}{Lemma}
\newtheorem{theorem}{Theorem}

\newenvironment{proof}{\trivlist\item[]\emph{Proof}:}%
{\unskip\nobreak\hskip 1em plus 1fil\nobreak$\Box$\parfillskip=0pt%
\endtrivlist}

\newcommand{\shortver}[1]{}

\newcommand{\shlongver}[2]{#2}

\newcommand{\cC}{{\cal C}}
\newcommand{\cE}{{\cal E}}
\newcommand{\cA}{{\cal A}}
\newcommand{\cB}{{\cal B}}
\newcommand{\cR}{{\cal R}}
\newcommand{\ra}{\mathrm{rank}}
\newcommand{\ssucc}{\mathrm{succ}}

\begin{document}
\title{New Data Structures for Orthogonal Range Reporting and Range Minima Queries}
\author{Yakov Nekrich\thanks{{Department of Computer Science, Michigan Technological University. Email {\tt yakov.nekrich@googlemail.com}}}}
\date{}
\maketitle
\begin{abstract}
  In this paper we present new data structures for two extensively studied variants of the orthogonal range searching problem.

  First, we describe a data structure that supports two-dimensional orthogonal range minima queries in $O(n)$ space and $O(\log^{\eps} n)$ time, where $n$ is the number of points in the data structure and $\eps$ is an arbitrarily small positive constant.  Previously known linear-space solutions for this problem require $O(\log^{1+\eps} n)$ (Chazelle, 1988) or $O(\log n\log \log n)$ time (Farzan et al., 2012).  A modification of our data structure uses space $O(n\log \log n)$ and supports range minima queries in time $O(\log \log n)$. Both results can be extended to support three-dimensional five-sided reporting queries.  

  Next, we turn to the four-dimensional orthogonal range reporting problem  and  present a  data structure that answers queries in optimal $O(\log n/\log \log n + k)$ time, where $k$ is the number of points in the answer. This is the first data structure that achieves the optimal query time for this problem.

  Our results are obtained by exploiting the  properties of three-dimensional shallow cuttings. 
\end{abstract}

\thispagestyle{empty}
\newpage
\setcounter{page}{1}

\section{Introduction}
\label{sec:intro}
Orthogonal range searching is a fundamental and extensively studied data structuring problem~\cite{McCreight85,Chazelle86,ChazelleG86,ChazelleG86a,ChazelleE87,Chazelle88,Overmars88,Chazelle90a,Chazelle90b,SubramanianR95,VengroffV96,AlstrupBR00,AlstrupBR01,Nekrich07isaac,Nekrich07,Nekrich07algorithmica,KarpinskiN09,NN12,ChanLP11,Chan13}. In this problem we store a set of multi-dimensional points $P$ in a data structure so that for an arbitrary axis-parallel rectangle $Q$ some information about points in $Q\cap P$ must be returned. Different variants of range searching queries have been studied by researchers: an orthogonal range reporting query asks for the list of all points in $P\cap Q$; an orthogonal  range emptiness query determines whether  $P\cap Q=\emptyset$, an orthogonal range counting query asks for the number of points in $P\cap Q$. In the range minima/maxima problem each point is assigned a priority  and we must return the point of smallest/highest priority in $P\cap Q$.

In this paper we study the orthogonal range reporting and range minima problems.
We improve the query time of linear-space range minima data structure
in two dimensions from $O(\log n\log \log n)$ to $O(\log^{\eps} n)$. Henceforth $n$ denotes the total number of points in the data structure and $\eps$ is an arbitrarily small positive constant. 
We also describe a data structure with optimal $O(\log n/\log\log n)$  query time for the four-dimensional orthogonal range reporting problem. 

\paragraph{Range Minima Queries.}
The best previously known trade-offs are listed in Table~\ref{tab:rangemin}. The study of compact data structures for range searching problems was initiated by Willard~\cite{Willard86} and Chazelle~\cite{Chazelle88}.  In the latter work, published over three decades ago,  Chazelle~\cite{Chazelle88} described an $O(n)$-space data structure that supports two-dimensional range minimum  queries\footnote{Range minima and range maxima problems are equivalent. In this paper we will talk about the range minima problem.} in $O(\log^{1+\eps} n)$ time.  The only improvement for an $O(n)$-space data structure was achieved by Farzan et al.~\cite{FarzanMR12} who  reduced the query time to $O(\log n \log \log n)$. 

Better query times for this problem can be achieved at a cost of increasing the space usage. Chan et al.~\cite{ChanLP11} described a data structure that uses $O(n\log^{\eps} n)$ space and answers queries in optimal $O(\log \log n)$ time. Another trade-off was achieved by Karpinski and Nekrich\cite{KarpinskiN09}; combining their data structure with the result of~\cite{Chan13}, we can obtain a data structure that uses $O(n (\log \log n)^3)$ space and answers queries in $O((\log\log n)^2)$ time. A five-sided range reporting query is a special case of three-dimensional orthogonal range reporting queries where a query is bounded on five sides. Range minima problem is closely related to the  special case of three-dimensional range reporting when the query range is bounded on five sides. All previous results, except for~\cite{FarzanMR12}, can be extended to support five-sided queries.

Range minima queries are to be contrasted with two-dimensional emptiness queries. In this problem we store a set of two-dimensional points; given an axis-parallel query rectangle $Q$, we must decide whether $Q\cap P=\emptyset$. Emptiness queries can be answered in $O(\log^{\eps}n)$ time using an $O(n)$-space data structure~\cite{ChanLP11}, or in $O(\log \log n)$ time using an $O(n\log \log n)$-space data structure~\cite{AlstrupBR00,ChanLP11}.  

In this paper we demonstrate that the gap between range emptiness and range maxima in two dimensions can be closed completely. We present a data structure that uses $O(n)$ space and answers range minima queries in $O(\log^{\eps} n)$ time. We also describe a data structure that uses $O(n\log\log n)$ space and answers queries in $O(\log\log n)$ time. Our results can be  extended to five-sided queries in three dimensions.

Our data structure employs the standard recursive grid approach frequently used in orthogonal range searching problems~\cite{AlstrupBR00,ChanLP11,FarzanMR12,KarpinskiN09,ChanNRT18}. The novel part of our method is a compact data structure supporting three-dimensional dominance queries for each recursive sub-structure of the grid. This data structure is based on the notion of a $t$-shallow cutting. We show that a shallow cutting can be "covered", in a certain sense,  by a set of rectangles. Every covering rectangle contains a small number of points and is unbounded along the third dimension. We can specify the relevant points by their positions in rectangles. This approach essentially reduces the problem of storing points in a recursive grid to the problem of storing points in a compact range tree aka the ball inheritance problem~\cite{Chazelle88,ChanLP11}. 
\begin{table}[tb]
  \centering
  \begin{tabular}{|l|c|c|c|}\hline
    Ref. & Space & Range-Minima &  Five-sided\\\hline
    \cite{Chazelle88} & $O(n)$ & $O(\log^{1+\eps} n)$ & $O(\log^{1+\eps} n+ k\log^{\eps}n)$\\
    \cite{FarzanMR12} & $O(n)$ & $O(\log n \log\log n)$ & - \\
    \cite{ChanLP11}  & $O(n\log^{\eps} n )$ & $O(\log\log n )$ & $O(\log\log n + k)$\\
    \cite{KarpinskiN09} + \cite{Chan13} & $O(n (\log\log n)^3)$ & $O(\log\log n)^2$ & $O((\log\log n)^2 + k\log\log n)$\\  \hline
    Our & $O(n)$ & $O(\log^{\eps} n)$  & $O((k+1)\log^{\eps}n)$ \\
    Our & $O(n\log\log n)$ &$O(\log\log n)$ & $O((k+1)\log\log n)$\\ \hline
  \end{tabular}
  \caption{Previous best and new results on range minima and five-sided range reporting problems.}
  \label{tab:rangemin}
\end{table}

\paragraph{Multi-Dimensional Range Reporting.}
Orthogonal range reporting queries can be answered in $O(\log \log n + k)$ time in two~\cite{Overmars88,AlstrupBR00} and three dimensions~\cite{Chan13}. By the lower bound for predecessor queries~\cite{PatrascuT06}, this query time is optimal. 
The best previously known four-dimensional data structure supports queries in $O(\log n + k)$ time.  According to the lower bound of \Patrascu, any data structure that consumes $O(n\mathtt{polylog}(n))$ space requires $\Omega(\log n/\log \log n)$ time to answer four-dimensional queries; this lower bound is also valid for emptiness queries. In this paper we  describe, for the first time, a data structure that achieves the optimal $O(\log n/\log\log n+k)$ query time for the four-dimensional range reporting problem. Henceforth $k$ denotes the number of points in the query range.

Previous solutions of this problem employed range trees to solve the orthogonal range reporting problem in four dimensions. To answer a query, we must navigate a  node-to-leaf path\footnote{Depending on the type of the query and the data structure, we may have to navigate along $O(1)$ different node-to-leaf  paths. For simplicity, we discuss the case of exactly one path.} in a range tree and answer a three-dimensional range reporting query in every node on that path.  By the lower bound for three-dimensional reporting, we have to  spend $\Omega(\log\log n)$ time in every visited node. The node degree of the range tree is bounded by $\log^{O(1)} n$; a higher node degree would lead to prohibitively high space usage because we must store a separate data structure for every range of node children.  Thus we must navigate along a path of  $\Omega(\log n/\log\log n)$ nodes and spend $\Omega(\log \log n)$ time in every node. For this reason  all previous methods need $\Omega(\log n)$ time.

Again, in this paper we achieve better query time by using $t$-shallow cuttings. Our solution is based on embedding a high-degree tree $T_0$ into the  range tree $T$ and storing $t$-shallow cuttings in the nodes of $T_0$. These $t$-shallow cuttings provide us with additional information and enable us to spend $o(\log \log n)$ time in every visited node of $T$ when a query is answered. In order to achieve the optimal query time, we use a sequence of embedded trees $T_0$, $T_1$, $\ldots$ with decreasing node degrees.

Throughout this paper $n$ will denote the total number of points in a data structure. The number of points in a sub-structure of a global structure will be sometimes denoted by $m$. 
We  assume w.l.o.g. that all point coordinates are positive integers bounded by $n$ and all points have different coordinates. The general case can be reduced to this case by applying the reduction to rank space technique.  Our results are valid in the standard RAM model. In this model we assume that the word size is $\Theta(\log n)$ and  that standard arithmetic operations can be performed on words in constant time. (In some cases our methods also make use of ``non-standard'' operations. However we can always implement these operations with table look-ups. The necessary look-up tables can be initialized in $O(n^{\eps})$ time.). The space usage is measured in words of $\log n$ bits, unless specified otherwise.

First we describe the linear-space data structure for five-sided three-dimensional range reporting. We explain how  the standard grid approach in Section~\ref{sec:recur}. We prove our result about ``covering'' a shallow cutting by rectangles in Section~\ref{sec:cover} and show how this covering can be used to obtain a compact dominance reporting data structure in Section~\ref{sec:slabdomin}. The same data structure can be easily modified to support two-dimensional range minima queries. Using the same approach, we obtain an  $O(n\log\log n)$-space data structure for range minima and five-sided reporting queries with $O(\log \log n)$ (resp. $O((k+1)\log\log n)$ query time; this result  is described in Section~\ref{sec:minima2}.  The data structure for four-dimensional orthogonal range reporting is presented in Section~\ref{sec:fourdim}.

\section{Recursive Grid}
\label{sec:recur}
We divide the $(x,y)$ grid into $\sqrt{n/\log^3 n}$ vertical slabs and $\sqrt{n/\log^3n}$ horizontal slabs, so that each slab contains $\sqrt{n\log^3 n}$ points.   For every slab we keep a data structure supporting three-dimensional dominance queries, that will be described in Section~\ref{sec:slabdomin}. This data structure uses $O(\log\log n)$ bits per point and answers queries in $O((k+1)\log^{\eps} n)$ time. The top data structure $D^t$ contains the $\log^{1+\eps} n$ points with smallest  $z$-coordinates from every cell. Last, we store a recursively defined data structure for every slab that contains $\Omega(\log^8 n)$ points. When the number of points in a slab does not exceed $O(\log^8 n)$, we keep all points in a data structure that uses $O(\log \log n)$ bits per point and supports queries in $O(1+k)$ time. This data structure can be constructed using standard techniques; see  Section~\ref{sec:rankspace}. 

We observe that an $O(n)$-word data structure supporting five-sided queries in $O(\log^{1+\eps}n+k\log^{\eps}n)$ time is already known~\cite{Chazelle88}. 
In Sections~\ref{sec:recur} - \ref{sec:slabdomin} we describe the data structure for \emph{$\log n$-capped} queries: we report all $k$ points in the query range if $k\le \log n$; if $k> \log n$, we return $NULL$. In the latter case, we can use the "slow" data structure of Chazelle and report all points in the query range in time  $O(\log^{1+\eps}n+k\log^{\eps}n)=O((1+k)\log^{\eps}n)$

A four-sided query (i.e., a 3-d query bounded on four sides) can be answered as follows. If a query range $[a,b]\times [0,h]\times[0,z]$ is entirely contained in one horizontal slab $H$, we answer the query using the data structure for $H$. If a query is contained in one vertical slab $V$, we answer the query using the data structure for $V$. Suppose the query intersects several horizontal slabs and several vertical slabs. Let $H_j$ denote the horizontal slab that contains $h$; let $V_a$ and $V_b$ denote vertical slabs that contains $a$ and $b$ respectively. 
A query is split into four parts, see Fig.~\ref{fig:grid} in Section~\ref{sec:figures}. The central part is aligned with slab boundaries, three other parts are contained in slabs $V_a$, $V_b$, and $H_j$ respectively. Let $a'$ denote the $x$-coordinate of the right boundary of $V_a$ and let $b'$ denote the $x$-coordinate of the left boundary of $V_b$. Let $h'$ denote the $y$-coordinate of the lower boundary of $H_j$. We can report points in $[a',b']\times [0,h']$ using the top data structure.  We ask dominance queries $[0,b]\times [0,h]$ and $[a,+\infty)\times [0,h]$ to slabs $V_a$ and $V_b$ respectively.  We ask a three-sided query $[0,h]\times [a',b']$ to horizontal slab $H_j$. A query to the top data structure and two dominance queries take $O(\log^{\eps}n)$ time. The total query time is $q(n)=O(\log^{\eps}n)+ q(\sqrt{n\log^3n})= O(\log^{\eps}n\log\log n)$.  

A five-sided query $[a,b]\times [h_1,h_2]\times [0,z]$ is processed in a similar way.  If a query is entirely contained in one horizontal or vertical slab, we answer the query using the data structure for that slab.  If  a query intersects several horizontal slabs and several vertical slabs, we split the query range into five parts. Let $V_a$ and $V_b$ denote the vertical slabs that contain $a$ and $b$ respectively; let $H_1$ and $H_2$ denote the horizontal slabs that contain $h_1$ and $h_2$. We answer three-sided queries $[a,+\infty)\times [h_1,h_2]$ and $[0,b]\times [h_1,h_2]$ on slabs $V_a$ and $V_b$ respectively. We answer two other three-sided queries $[a,b]\times [h_1,+\infty)$ and $[a,b]\times [0,h_2]$ on slabs $H_1$ and $H_2$.  The central part of the query can be answered using the top data structure $D^t$. The total query time is dominated by three-sided queries, $q(n)=O(\log^{\eps} n\log\log n)$. We can reduce the time to $O(\log^{\eps} n)$ by replacing $\eps$ with an arbitrary $\eps'< \eps$ in the above construction.

Let $S(n)$ denote the space usage of our structure in bits. The top data structure can be implemented using e.g.~\cite{KarpinskiN09} and requires $O((n/\log^3 n)\log^{2+\eps})=o(n)$ bits.  Dominance data structures use $O(n\log\log n)$ bits because a dominance data structure consumes $O(\log\log n)$ bits per point and each point is kept in two slabs. Hence $S(n)= O(n\log\log n) +2\sqrt{n/\log^3 n}S(\sqrt{n \log^3 n})$. We set $c(n)=S(n)/n$ and divide both parts of the previous equality by $n$. Thus  $c(n)=O(\log\log n) + 2c(\sqrt{n\log^3 n})$. The latter recursion can be resolved to $c(n)=O(\log n)$.  Hence $S(n)=O(n\log n)$ and the data structure  uses $O(n)$ words of $\log n$ bits.   

\begin{theorem}
  \label{lemma:capped5side}
  There exists a data structure that uses $O(n)$ words of space and supports five-sided queries in $O((k+1)\log^{\eps}n)$ time. 
\end{theorem}

\section{Covering of a Shallow Cutting}
\label{sec:cover}
Our main tool in designing a compact dominance data structure is shallow cuttings. 
A $t$-shallow cutting for a set $P$ is a collection of $O(n/t)$ cells such that each cell $C$ is a rectangle of the form $[0,a]\times [0,b]\times [0,c]$. Furthermore each cell contains at most $2t$ points from $P$ and every point $q$ in 3-d space that dominates at most $t$ points from $P$ is contained in some cell(s).  The list of all points from $P$ in a cell $C$ is called a conflict list of $C$ and denoted $list(C)$.  For a  cell $C=[0,a]\times [0,b]\times [0,c]$, the point $(a,b,c)$ is called the corner of $C$.

In this section we show that conflict lists of all cells in a $t$-shallow cutting can be almost covered by 3-d boxes unbounded in $z$-direction. The conflict list of each cell is contained in $O(d^3)$ boxes where $d$ is a parameter that does not depend on $t$. There can be a small number of points that is not contained in these $O(d^3)$ boxes. However the total number of such points for all conflict lists is $O(|P|/d)$. 
\begin{theorem}
  \label{theor:shallowpart}
Let $\cC$ denote a $t$-shallow cutting of a three-dimensional point  set $P$, $|P|=m$. For any integer $d>0$ there exists a subset $P'$ of $P$ and a set of three-dimensional rectangles $\cR=\{\,R_1,R_2,\ldots, R_s\,\}$, such that
\begin{description}
\item[(a)] 
  $|P'|=O(m/d)$
\item[(b)]
  Rectangles $R_i$ are unbounded along the $z$-axis.
\item[(c)]
  The conflict list of any cell, except for points from $P'$, is contained in $O(d^3)$ rectangles from $\cR$, \[list(C_i)\cap (P\setminus P')\subseteq (list(C_i)\cap R_{i_1})\cup (list(C_i)\cap R_{i_2})\cup\ldots\cup (list(C_i)\cap R_{i_g}\] for $g=O(d^3)$.
\item[(d)]
  Each rectangle contains $O(t\cdot d^4)$ points of $P$. 
\end{description}
\end{theorem}
\begin{proof}
 We represent $P'$ as a union of three sets $P'=P_1\cup P_2\cup P_3$.  $P_1$ is the subset containing all points from $P$ that are stored in conflict lists of at least $d$ different cells for a parameter $d$. The number of points in $P_1$ is at most $O(n/d)$.  Next we construct the set $P_2$ and the set of rectangles $\cR$, so that conditions (b) and (c) are satisfied. Finally we will remove some rectangles from $\cR$ and construct $P_3$, so that condition (d) is satisfied. For simplicity we sometimes do not distinguish between a rectangle or a point and its projection on the $xy$-plane. 

 \paragraph{Staircases, Regions, and Neighborhoods}
 Consider the corners $c_i$ of all cells $C_i$ in a shallow cutting.
 We can assume that no $c_i$ is dominated by another corner $c_j$ (if this is the case, $C_i$ is contained in $C_j$ and we can remove the cell $C_i$ from the shallow cutting).
 
Let $M$ be the set of projections of corners onto the $xy$-plane. We decompose $M$ into maximal layers (layers of maxima) $M_i$: $M_1$ is the set of maximal points\footnote{A point $p$ in a set $M$ is maximal if $p$ is not dominated by any other point in $M$.} of $M$. $M_2$ is the set of maximal points of $M\setminus M_1$ and $M_i$ for $i> 2$  is the set of maximal points in $M\setminus (\cup_{j=1}^{i-1} M_j)$. Thus every point on $M_i$ is dominated by some point on $M_{i-1}$ and no point on $M_i$ is dominated by another point on $M_i$. We connect points of $M_i$ by alternating horizontal and vertical segments; the resulting polyline will be called the \emph{staircase} of $M_i$. See Fig.~\ref{fig:shcut-cover} in Section~\ref{sec:figures}.

We visit corners of $\cC$ on a layer $M_i$ in the decreasing order of their $z$-coordinates.
Let $c_j$ denote the $j$-th visited corner on $M_i$.  We shoot a horizontal ray in $-x$ direction from $\pi(c_j)$ until it hits either a ray of a previously visited corner $c_l$, $l<j$, or the staircase of $M_{i+d}$.  We also shoot a vertical ray from $\pi(c_j)$ until
it hits either a ray of a previously visited corner $c_{l'}$, $l'<j$, or the staircase of $M_{i+d}$. We will call the polygon bounded by two ray from $\pi(c_j)$,  the rays from $c_l$ and $c_{l'}$ and a portion of the staircase $M_{i+d}$ the \emph{region} of $c_j$.  

If the region of $c_j$ contains at least $d^2$ corners of $M_j$ for all  $j$ such that $i<j \le i+d$, we add all points of $list(C_j)$ to $P_2$ and say that the region of $c_j$ is empty. Otherwise we divide the region of $c_j$ into at most $d$ rectangles, called rectangles of $c_j$ (or rectangles associated to $c_j$).  See Fig.~\ref{fig:shcut-cover} in Section~\ref{sec:figures}. The set of rectangles $\cR$ consists of all rectangles associated to corners of $\cC$.

The \emph{neighborhood} of a corner $c_j$ is defined as follows. We shoot a horizontal ray in $-x$ direction until we either (i) hit the boundary of an empty region or (ii) cross the boundaries of $d$ non-empty regions or (iii) hit the staircase of $M_{i+d}$. We also shoot a vertical ray in $-y$ direction until we either (i) hit the boundary of an empty region or (ii) cross the boundaries of $d$ non-empty regions or (iii) hit the staircase of $M_{i+d}$.  We call all rectangles whose boundaries are crossed by the vertical and the horizontal rays from $c_j$ the \emph{neighbors} of $c_j$.

\begin{lemma}
  There are at most $O(d^3)$ rectangles associated to  neighbors of a corner $c_j$.
\end{lemma}
\begin{proof}
  Each corner has at most $2d$ neighbor regions. Every region is divided into $O(d^2)$ rectangles. 
\end{proof}

\begin{lemma}
  \label{lemma:cov1}
  For any point $p$ in $list(C_j)$, either $p\in P_1\cup P_2$ or $p$ is contained in some rectangle $R_{j_f}$ associated to a neighbor of $c_j$. 
\end{lemma}
We can bound the number of points in conflict lists of corners with empty regions.
\begin{lemma}
  \label{lemma:p2}
  The number of points in a set $P_2$ is bounded by $O(m/d)$. 
\end{lemma}
Proofs of Lemma~\ref{lemma:cov1} and Lemma~\ref{lemma:p2} can be found in Section~\ref{sec:coverlemmas}.
\paragraph{Constructing the set $P_3$.}
In the final step, we exclude rectangles that contain too many points from $P$ from the set $\cR$.
Let $list(\cC)=\cup_{C_j\in \cC}list(C_j)$ denote all points from $P$ contained in the cells of $\cC$.  For any rectangle $R\in\cR$, $|R\cap list(\cC)|\le t\cdot d^2$: a rectangle $R$ is a subset of the non-empty region associated to  some corner $c_j$. Hence every rectangle intersects projections of at most $d^2$ different cells.  Therefore $R\cap list(\cC)\le t\cdot d^2$ for all $R\in \cR$.

If $|R\cap P|\ge t\cdot d^4$ for some rectangle $R\in \cR$, we add all points of $R\cap list(\cC)$ to $P_3$ and remove the rectangle $R$ from $\cR$. We can show that $|P_3|=O(m/d)$:   Similar to Lemma~\ref{lemma:p2}, we assign $d$ dollars to each point of $P$ and assume that the cost of inserting a point into $P_3$ is $d^2$ dollars. If points from $R\cap list(\cC)$ are added to $P_3$, then we charge the cost to $P\cap R$. The cost is evenly distributed among all points of $P\cap R$. Since $|P\cap R|\ge d^4\cdot t$ and $|list(\cC)\cap R|\le d^2\cdot t$,  we charge no more than  $1$ dollar to each point. 

All rectangles associated to corners of a fixed maximal layer $M_i$ are disjoint.  A planar point between the staircase of $M_i$ and the staircase of $M_{i+1}$ can be covered by at most one rectangle associated to some corner of $M_j$ for every $j$, $i-d< j\le i$. Therefore $\pi(p)$ for any point $p\in P$ is contained in at most $d$ rectangles from $\cR$. 
Thus each point of $P$ is charged at most $d$ times. Hence the total cost of creating $P_3$ is $O(m\cdot d)$ dollars, where $m$ is the total number of points in $P$. Since the cost of inserting a point into $P_3$ is $d^2$,  $P_3$ contains $O(m/d)$ points.

Summing up, $P'=P_1\cup P_2\cup P_3$ contains $O(m/d)$ points. For every cell $C_j$ of $\cC$, $list(C_j)\cap (P\setminus P')$ is covered by $O(d^3)$ rectangles from $\cR$ and every $R\in \cR$ contains $O(t\cdot d^4)$ points of $P$. 
\end{proof}

\section{Dominance Queries in a Slab}
\label{sec:slabdomin}
Now we describe the compact data structure that supports capped dominance range reporting queries.  By a slight misuse of notation, in this Section  $P$ will denote the set of points in a slab.

Let $P$ denote the set of points in a slab $u$ and let $m=|P|$. We  construct a $t$-shallow cutting with $t=\log^6 n$ and  apply Theorem~\ref{theor:shallowpart} with $d=\log n$; the subset $P'\subset P$ and the set of rectangles $\cR$ are as defined in Theorem~\ref{theor:shallowpart}.  We  keep  $P'$ in a data structure from~\cite{Chan13} that uses $O(\log n)$ bits per point and answers queries in $O(\log \log n + k)$ time. Let $P_j=list(C_j)\setminus P'$ denote the set of points in the conflict list of the cell $C_j$ that are not in $P'$.  We construct a dominance data structure for points in $P_j$. Points in $P_j$ are reduced to the rank space, so that the cell  data structure uses $O(\log t)=O(\log \log n)$ bits per point.  A $(\log n)$-capped dominance query is answered as follows: we find the cell $C$ that contains $q$, reduce $q$ to the rank space of $C$, and report all points in $P_j$ that are dominated by $q$. We also query the data structure for $P'$ and report all points in $P'$ that are dominated   by $q$.  If $q$ is not contained in any cell of the $t$-shallow cutting, then $q$ dominates at least $t=\log^5 n$ points and we can return $NULL$.

It remains to show how the points can be "decoded", i.e., how to obtain the coordinates of a point from its coordinates in the rank space of $P_j$.  We also need to show how the query point $q$ can be transformed into the rank space.  For each cell $C_j$ of the shallow cutting we keep the list $rlist(C_j)$ of rectangles $R_{i_1}$, $R_{i_2}$, $\ldots$, such that $list(C_j)\setminus P'$ is contained in these rectangles; for every rectangle we store the global coordinates of its endpoints.  By Theorem~\ref{theor:shallowpart} $rlist(C_j)$ consists of $O(d^3)$ rectangles.  We can identify a point $p$ in $list(C_j)\setminus P'$  by specifying the rectangle $R_{i_p}$ that contains $p$ and its $x$-rank in the rectangle $R_{i_p}$ (the $x$-rank of a point $p$ in a rectangle $R$ is the number of points in $R$ to the left of $p$). Since the list $rlist(C_j)$ consists of $O(d^3)$ rectangles, the rectangle $R_{i_p}$ can be identified by its position in $rlist$ using  $O(\log d)$ bits.  Each rectangle $R\in \cR$ contains of $O(t\cdot d^4)$ points. Hence we can store the $x$-rank of $p$ in $R_{i_p}$ using $O(\log t + \log d)$ bits. Thus $p$ is represented using $O(\log t+ \log d)=O(\log \log n)$ bits.    Hence every rectangle $R\in \cR$ contains a poly-logarithmic number of points from the global set of points. 

Every rectangle $R\in \cR$ is unbounded along the $z$-axis.
Hence we can retrieve the point coordinates by answering a special kind of two-dimensional queries, further called \emph{capped range selection} queries.   A capped range selection query $(Q,v)$ for a two-dimensional rectangle $Q=[a,b]\times [c,d]$ and an integer $v\le \log^{10} n$ returns the point with $x$-rank $v$ in the rectangle $Q$, i.e., the $v$-th leftmost point in $Q$. Let $P_{glob}$ denote the global set of points.  Slabs are created by dividing the structure on previous recursion level along the $x$- or $y$-axis. Therefore $P_{glob}\cap R= P\cap R$ for every $R\in \cR$.  Hence 
we need only one instance of the capped range selection data structure for all recursive slabs.

Capped range selection can be viewed as a generalization of the range successor problem~\cite{NN12} and can be solved using compact range trees and some auxiliary data structures.  We will show in Section~\ref{sec:smallsel} that capped selection queries can be answered in $O(\log^{\eps}n)$ time using an $O(n)$ space data structure. 
However, we have to decompose each rectangle from $\cR$ into $O(\log^{1+\eps} n)$ parts.
Therefore we need $O(\log^{2+\eps} n)$ additional bits per rectangle. A detailed description is provided in Lemma~\ref{lemma:smallsel} in Section~\ref{sec:smallsel}.

Using Lemma~\ref{lemma:smallsel}, we can retrieve the coordinates of any point $p$ in $C_j$  in $O(\log^{\eps} n)$ time: we know the rectangle $R\in \cR$ that contains  $p$ and we know the $x$-rank $\ell$ of  $p$ in $R$. Hence we can return the  global coordinates of $p$ in $O(\log^{\eps} n)$ time by answering a query $(R,\ell)$. If we can retrieve a point from cell $C$ in time $O(\log^{\eps} n)$, we can also reduce the query $q$ to the rank space of $C$ within the same time; see Section~\ref{sec:rankspace}.

The conflict lists of all cells contain $O(m)$ points. Data structures for all cells of the shallow cutting use $O(m\log\log n)$ bits. For each cell we also store the list of rectangles and spend $O(\log^{2+\eps} n)$ bits per rectangle (for the capped selection data structure). In total we need $O(d^3\log^{2+\eps} n)=o(t)$ bits for every cell and $O(m)$ bits for all cells. The data structure for the subset $P'$ uses $O((m/\log n)\log n)=O(m)$ bits.  To identify a cell of the shallow cutting that contains the query point, we need a data structure that supports point location queries on a planar orthogonal subdivision of size $O(m/t)$~\cite{Chan13}. This data structure uses $O(m/\log^5n)$ bits. Hence our dominance data structure uses $O(m\log \log n)$ bits.  We need only one instance of range selection data structure for all slabs; in this section we use the variant that consumes $O(n\log n)$ bits of space. Hence the total space usage of the global data structure is not affected. 

We also use range selection to access points stored in the slabs of size $O(\log^8 n)$ of the grid structure. The number of points in every slab that is not divided further is $\Omega(\log^4 n)$. Since every point is stored in $O(\log n/\log\log n)$ recursive structures, the number of slabs is bounded by $O(n/\log^3 n)$. Since the range selection data structure requires  $O(\log^{2+\eps} n$ bits per rectangle (i.e., per slab),  the space usage of the range selection data structure is $O(n)$.

\paragraph{Two-Dimensional Range Minima Queries.}
Our data structure for five-sided queries can be easily modified to support range minima queries. We can adjust the slab dominance data structure, described in this section, so that it supports 2-d dominance minima queries. We use the same shallow cutting, but construct a data structure supporting 2-d dominance minimum queries for every cell.  We replace the top data structure with a data structure from~\cite{ChanLP11}.  This data structure contains
$O(n/\log n)$ points and uses $O(n)$ space (on the top recursion level). 
We can decompose a  four-sided query to $O(\log\log n)$ dominance queries and $O(\log\log n)$ minima queries on top data structures as described in Section~\ref{sec:recur}. Thus the answer to a (four-sided) range minima query in 2-d is the minimum of the answers to $O(\log\log n)$ local queries. The total query time $O(\log^{\eps}n\log\log n)$. We can get rid of the $\log\log n$ term by replacing $\eps$ with an arbitrary $\eps'< \eps$.

\section{Four-Dimensional Dominance Range Reporting}
\label{sec:fourdim}
In this section we describe a data structure that answers dominance range reporting queries in $O(\log n/\log \log n+k)$ time. We start by explaining why previous methods require $\Omega(\log n)$ time. Then
we describe the main idea of our approach and show how it can be used to reduce the query time by a factor $O(\sqrt{\log\log n})$. Then we will present a complete solution.

\paragraph{Range Trees.}
Four-dimensional queries are reduced to three-dimensional queries using a range tree.  A range tree $T$ stores the points of the input set $P$ sorted by their fourth coordinate. In every internal node $u$ of the range tree, we keep a set of points $S(u)$; $S(u)$ contains all points stored in the leaf descendants of $u$.  In every node $u$ we  keep a  data structure supporting  three-dimensional  dominance queries. Let $\gamma$ denote the node degree of the range tree. For any interval $[a,b]$, we can identify $O(\gamma \cdot \log_{\gamma} n)$ nodes $u_i$ of the range tree, such that $p.z'\in [a,b]$ if and only if $p\in S(u_i)$. Furthermore nodes $u_i$ can be divided into $O(\log_{\gamma}n)$ groups, such that nodes in the same group are siblings.

A three-dimensional dominance query can be answered by locating  a point in a planar orthogonal subdivision~\cite{Nekrich07,Afshani08,Chan13}.  Orthogonal point location queries can be answered in $O(\log \log n)$ time~\cite{Chan13}. If the range tree $T$ is a binary tree,  we can answer four-dimensional dominance queries in $O(\log n \log \log n + k)$ time.  By increasing the node degree $\gamma$ to $O(\log^{\eps}n)$, we can reduce the query time to $O(\log n + k)$. 
Unfortunately it appears that further improvement in query time is not possible with this approach:  $O(\log \log n)$ query time is optimal for both planar orthogonal point location and   three-dimensional reporting queries.  This lower bound is valid for any data structure with pseudo-linear space usage and follows from the lower bound for the predecessor problem~\cite{BeameF02,PatrascuT06}. Increasing the node degree is also not feasible: every point must be stored in $\Omega(\gamma^2)$ three-dimensional data structures. Hence if $\gamma=\log^{\omega(1)}n$,  the total space usage would be  prohibitively  high.

\paragraph{Better Query Time.}
In order to improve the query time we store additional information for selected nodes in the range tree. Our range tree $T$ has node degree $\gamma=\log^{\eps}n$. We embed a tree $T^0$ with node degree $\rho_0=\gamma^{\alpha_0}$ where $\alpha_0=\log \log n$ into $T$. Nodes of $T^0$, further called $0$-nodes, correspond to nodes of $T$ with depth that is divisible by $\log \log n$. For every range $1\le f\le l \le \rho_0$ we define $S(u,f,l)=\cup_{j=f}^l S(u_j)$ where $u_j$ denotes the $j$-th child of a node $u$. Sets $S(u,\cdot,\cdot)$ will be called the node ranges of $u$. Let $t_0=\rho_0^4$. We construct a $2t_0$-shallow cutting $\cC(u,i,j)$ for each $S(u,i,j)$ and for every internal node with height at least $2$ in the tree $T^0$. We also construct a $2t_0$-shallow cutting $\cE(v,i,j)$ for each set $S(v,i,j)=\cup_{l=i}^j S(v_l)$ where $v$ is a node in $T$ and $v_l$ is the $l$-th child of $v$.  Finally, we construct a $t_h$-shallow cutting $\cC'(E_j)$ for each cell $E_j$ of $\cE(v,i,j)$, where $t_h=\gamma^{O(1)}$. 

The set $S(u,i,j)$ of a $0$-node $u$ can be represented as a union of $O(\log \log n)$ non-overlapping sets $S(v_f,i_f,j_f)$ where each node $v_f$ is "between" the node $u$ and its children in $T^0$. In other words,each $v_f$ is a  descendant of $u$ and an ancestor of at least one $u_l$, $i\le l\le j$. We will say that sets $S(v_f,i_f,j_f)$ are a \emph{canonical decomposition} of $S(u,i,j)$. 
\begin{lemma}[\cite{Nekrich20}]
  \label{lemma:contain}
Let $\cA$ be an $f$-shallow cutting for a set $S$ and let $\cB$ be an $(f')$-shallow cutting for a set $S'\subseteq S$ so that $f'\ge 2f$. Every cell $A_i$ of $\cA$ is contained in some cell $B_j$ of $\cB$.   
\end{lemma}
Consider a set $S(u,i,j)$ and its canonical decomposition $S(u,i,j)=\cup_f S(v_f,i_f,j_f)$.  By Lemma~\ref{lemma:contain}, a cell of $C_j$ of $\cC(u,i,j)$ is contained in some cell  $E_{j_f}$ of $\cE(v_f,i_f,j_f)$. For each cell $C_l$ of a shallow cutting $\cC(u,i,j)$  and for every set in the canonical decomposition of $S(u,i,j)$   we store a pointer to a cell $E_{l_f}\in \cE(v_f,i_f,j_f)$ containing $C_l$.

\begin{lemma}[\cite{Chan13}]
  \label{lemma:domin3d-fast}
  There exists a data structure that answers point location queries in a two-dimensional orthogonal subdivision of a $U\times U$ grid by  $O(m)$ rectangles in time $O(\min(\sqrt{\log_{\gamma} m},\log\log_{\gamma} U))$.   
\end{lemma}
\begin{proof}
  Using the predecessor data structure, we can reduce the point location problem to the special case when point coordinates are bounded by $O(m)$. Using the result of Chan~\cite{Chan13} we can answer point location queries on an $O(m)\times O(m)$ grid in $O(\log\log_{\gamma} m)$ time\footnote{The data structure described in~\cite{Chan13} supports queries in $O(\log\log m)$ time. A straightforward extension of this data structure supports queries in $O(\log\log_{\gamma} m)$ time; see e.g.,\cite{ChanNRT18}.} Predecessor queries can be answered in time $O(\sqrt{\log_{\gamma} m})$. Hence the total query time is $O(\sqrt{\log_{\gamma} m}+\log\log_{\gamma} m)=O(\sqrt{\log_{\gamma} m})$.  
\end{proof}

Now a query can be answered as follows. Suppose that we must report all points dominated by $q=(q_x,q_y,q_z,q_{z'})$. We visit the nodes $u$ of $T^0$ on the path from the root to $q_{z'}$. In every visited node $u$ we identify the canonical set $S(u,i,j)$ and find the cell $C_l$ of the shallow cutting $\cC(u,i,j)$ that contains $q'=(q_x,q_y,q_z)$.  See Fig.~\ref{fig:rangetree4d} for an example. Consider the canonical decomposition  $S(u,i,j)=\cup_f S(v_f,i_f,j_f)$. For every $f$, we visit the cell $E_{l_f}$ of $\cE(v_f,i_f,j_f)$ that contains $C_l$.  We identify the cell $C'_f$ of  the shallow cutting $\cC'(E_{l_f})$  that contains $q$. Every cell of $\cE(v_f,i_f,j_f)$ contains  $\gamma^{O(\log \log n)}$ points. Hence, we can answer a point location query and find the cell $C'_f$ in $O(\sqrt{\log \log n})$ time by Lemma~\ref{lemma:domin3d-fast}. Since $C'_f$ contains $O(\log^2 n)$ points, we can answer a three-dimensional query on $C'_f$ in $O(1+k)$ time; see Section~\ref{sec:rankspace}, Lemma~\ref{lemma:small0}. If $q'$ is not contained in any cell of $C'$, then the query range contains  $k\ge \log^2 n$ points and we can answer the query using some of the previously known data structures in time $O(\log n+ k)=O(k)$.

Our procedure visits $O(\frac{\log n}{(\log \log n)^2})$ nodes of $T^0$ and $O(\frac{\log n}{\log\log n})$ nodes of $T$. 
We spend $O(\log \log n)$ time in every visited node of $T^0$ and $O(\sqrt{\log\log n})$ time in every visited node of $T$.
Hence the total query time is $O(\log n/\sqrt{\log\log n})$. 

\paragraph{Optimal Query Time.}
In order to further improve the query time, we embed a sequence of subtrees $T^i$ into $T$. Node degrees of these subtrees decrease exponentially. As above, we keep shallow cuttings for node ranges in every tree. A shallow cutting in a node range  of $T^i$ provides us with a hint (via Lemma~\ref{lemma:contain}) that speeds up the search in $T^{i+1}$.

Let $\alpha_0=\log \log n$ and $\alpha_i=(\alpha_{i-1})^{1/2}\log^2 \alpha_{i-1}$.  If $u$ is a node in $T^i$ and $v$ is its child in $T^i$, then $\alpha_i$ is the distance between $u$ and $v$ in $T$. Thus every node of $T^i$ corresponds to subtree of height $\alpha_i$ in $T$.  We set $\rho_i=\gamma^{\alpha_i}$ and $t_i=(\rho_i)^4$. We choose $h$ so that $\alpha_h=\Theta(1)$.  To avoid clumsy notation, we assume that $\alpha_i$ divides $\alpha_{i-1}$ for $i\ge 1$ and $\alpha_0$ divides the  height of $T$.  Nodes of $T^i$ will be called $i$-nodes.

For every node $u$ of $T^0$, $0\le i\ge h$, and for every range of children $S(u,l,r)=\cup_{j=l}^r S(u_j)$, we construct a $t_i$-shallow cutting $\cC'(u,l,r)$.
For every node $u$ of $T^i$, $1\le i\le h$, and for every range of children $S(u,l,r)$, we construct a $2t_{i-1}$-shallow cutting $\cE(u,l,r)$. For each cell $E$ of $\cE(u,l,r)$ we construct a $t_i$-shallow cutting $\cC'(E)$.  Consider a canonical decomposition of $S(u,l,r)$ into $S(v_f,l_f,r_f)$,  $S(u,l,r)=\cup S(v_f,l_f,r_f)$, where $u$ is a node of $T^0$ and $v_f$ are nodes of $T^1$.    For each cell $C'$ of $\cC'(u,l,r)$ and for every set $S(v_f,l_f,r_f)$ in the canonical decomposition,  we keep the pointer to a cell $E_f$ of $\cE(v_f,l_f,r_f)$ such that $E_f$ contains $C$.     Consider a canonical decomposition of $S(u,l,r)$ into $S(v_f,l_f,r_f)$,  $S(u,l,r)=\cup_f S(v_f,l_f,r_f)$, where $u$ is a node of $T^i$ for some $i$, $1\le i< h$ and $v_f$ are nodes of $T^{i+1}$. For each cell $C'$ of $\cC'(E)$, where $E$ is a cell of $\cE(u,l,r)$,  and for every set $S(v_f,l_f,r_f)$ in the canonical decomposition,  we keep the pointer to a cell $E_f$ of $\cE(v_f,l_f,r_f)$ such that $E_f$ contains $C'$.

Consider a canonical decomposition of $S(u,l,r)$ into $S(v_f,l_f,r_f)$,  $S(u,l,r)=\cup_f S(v_f,l_f,r_f)$, where $u$ is a node of $T^h$.
For each cell $C'$ of $\cC'(E)$, where $E$ is a cell of $\cE(u,l,r)$,  and for every set $S(v_f,l_f,r_f)$ in the canonical decomposition of $S(u,l,r)$,  we keep the pointer to a cell $C_f$ of $\cC(v_f,l_f,r_f)$ such that $C_f$ contains $C'$. We remark that the nodes in the canonical decomposition of $S(u,l,r)$ are the nodes of $T$. We can adjust the constant $\gamma=\log^{\eps}n$ in such way, that $t_h=\log^3 n/2$. Hence, by Lemma~\ref{lemma:contain}, each cell $C'$ is contained in some $C_f$. 

A query is answered as follows.  The set $\pi^i$ consists of all $i$-nodes $u$ on the path from the root  to $q_{z'}$, such that the height of $u$ is at least $\ell=3\log\log n$. A query is processed in $h+1$ stages. During stage $i$ we visit nodes on $\pi^i$; for every node find the cell of $\cC'(E)$ that contains $q'=(q_x,q_y,q_z)$.   \\
{\tt Stage 0.}
We visit nodes of $T^0$ on the path $\pi^0$. In every visited node $u$ we find the canonical set $S(u,l,r)$ in the canonical decomposition of $[0,q_{z'}]$. Next we find the cell $C$ of $\cC'(u,l,r)$ that contains $q$. For every set $S(v_f,l_f,r_f)$ in the canonical decomposition of $S(u,l,r)$, we visit the cell $E$  of $\cE(v_f,l_f,r_f)$ that contains $C$. The we locate the cell $C'$ of $\cC'(E)$ that contains $q$.\\
{\tt Stage $i$, $1\le i\le h$.}
Suppose that we already know the cell $C'$ of $\cC'(E)$ that contains $q'$ in  every $i$-node $u$ on $\pi^i$ for some $i$, $1\le i< h$. For every node range $S(u,l,r)$ we consider its decomposition into $(i+1)$-nodes $S(u,l,r)=\cup S(v_f,l_f,r_f)$. For each $\cE(v_f,l_f,r_f)$ we  visit the  cell $E_f$ that contains $C'$; then we locate the cell $C'_f$ of  $\cC'(E_f)$ that contains $q$. \\
{\tt Final Step.}
Suppose that we already know the cell $C'$ of $\cC'(E)$ that contains $q'$ in  every $h$-node $u$ on $\pi^h$. For every node range $S(u,l,r)$ we consider its decomposition into $(i+1)$-nodes $S(u,l,r)=\cup S(v_f,l_f,r_f)$. For each $\cC(v_f,l_f,r_f)$ we  visit the  cell $C_f\in \cC(v_f,l_f,r_f)$ that contains $C'$. Finally we report all points in $list(C_f)$ that are dominated by $q'$.

Now we analyze the query time. 
To simplify the notation let  $\lambda(n)=\log n/\log \log n$ and let $t_0=n$. Our method visits $O(\lambda(n)/\alpha_i)$ $i$-nodes. The time spent in a visited node is dominated by the time needed to answer a point location query on a $O(t_{i-1})$ rectangles. Using the result of Chan~\cite{Chan13} we spend $O(\log \log n)$ time in each $0$-node and $O(\lambda(n))$ time in all $0$-nodes.  By Lemma~\ref{lemma:domin3d-fast}, we spend $O(\sqrt{\alpha_{i-1}})$ time in every $i$-node where $1\le i\le h$. Hence the total time in all $i$-nodes is $O(\lambda(n)\sqrt{\alpha_{i-1}}/\alpha_{i})=O(\lambda(n)/\log^2(\alpha_{i-1}))$. We can show that  $\sum_{i=0}^h \frac{1}{\log^2(\alpha_i)}=O(1)$; see Section~\ref{sec:analysis4d}. Hence, the total time that we need to locate $q'$ in the shallow cuttings of all relevant nodes is $O(\lambda(n))$.  A three-dimensional query on a cell $C_f$ of a $t$-shallow cutting takes time $O(1)$ (ignoring the time to report points). 

One technicality still needs to be addressed. We must consider the nodes $u$ on the path from the root to $q'_z$, such that the height of $u$ is less than $\ell$. We answer a three-dimensional query in every such node in $O(\log \log n)$ time. Since the number of such nodes is $\ell=O(\log\log n)$, the total query cost increases by a negligible term $O((\log \log n)^2)$.

All auxiliary shallow cuttings $\cE(\cdot)$ and $\cC'(\cdot)$ use linear space:  consider a node $u$. For every node range $S(u,l,r)$ we store two shallow cuttings that have  $O(|S(u,l,r)|/t_{i})$  $(|S(u,l,r)/t_{i-1})$ cells respectively. We store $\alpha_i/\alpha_{i+1}< \alpha_i$ pointers for each cell of $\cC'(E)$.  Since each point of $S(u)$ occurs in $\rho_i^2$ node ranges the total space used by all shallow cuttings associated to the $i$-node $u$ is  $O(m\cdot (\alpha_i/\rho^2_i))=O(m/\rho_i)$ where $m$ is the number of points in $S(u)$. Thus shallow cuttings in all nodes of $T^i$ consume $O((n\log n)/\rho_i)$ space.  Since $\sum_{i=0}^h (1/\rho_i)\log n=O(1)$, all additional shallow cuttings consume $O(n)$ space. Finally we can store points in the cells of  shallow cuttings $C_i(u,l,r)$ in $O(n\log^{\eps} n)$ words of space using the method of~\cite{Nekrich20}. Hence the total space usage is $O(n\log^{\eps}n)$. 
\begin{theorem}
  \label{theor:opttime}
  There exists an $O(n\log^{\eps} n)$-word data structure that answers four-dimensional dominance range reporting queries in $O(\log n/\log\log n+ k)$ time. 
\end{theorem}

We can use the same data structure to support five-sided four-dimensional queries, i.e., four-dimensional queries that are bounded on five sides.    Using standard techniques this result can be extended to a data structure that uses $O(n\log^{3+\eps}n)$ space and answers arbitrary four-dimensional orthogonal range reporting queries in $O(\log n/\log\log n+ k)$ time.
\begin{theorem}
  \label{theor:opttime2}
  There exists an $O(n\log^{3+\eps} n)$-word data structure that answers four-dimensional dominance range reporting queries in $O(\log n/\log\log n+ k)$ time. 
\end{theorem}

It is possible to reduce the space usage of the data structure in Theorem~\ref{theor:opttime2} to $O(n\log^{2+\eps} n)$ words using the lopsided grid approach from~\cite{ChanLP11}. \shlongver{This result is omitted from this paper for space reasons and will be presented in the full version of this paper.}{This result is presented in Section~\ref{sec:space4d}.}

\appendix
\section{Reduction to Rank Space and Range Reporting on a Small Set of Points}
\label{sec:rankspace}
We can reduce an orthogonal range searching problem on  a set $P$ of $m$ points to the special case when all point coordinates are positive integers bounded by $m$~\cite{GabowBT84,AlstrupBR00}. This can be achieved by replacing every point coordinate by its rank.
In the case of three-dimensional points every point $p=(p.x, p.y, p.z) $ in a set $P$ is
replaced with $p'=(\ra(p.x,S_x),\ra(p.y,S.y), \ra(p.z,S_z))$, where $P_x$, $P_y$, and $P_z$ denote the sets of $x$-, $y$-, and $z$-coordinates of points in $P$. For any point $p$ we have: 
\[ p\in [a,b]\times [c,d]\times [e,f] \Leftrightarrow p'\in [a',b']\times [c',d']\times [e',f']\]  where
$a'=\ra(\ssucc(a,P_x),P_x)$, $c'=\ra(\ssucc(c,P_y),P_y)$, $e'=\ra(\ssucc(e,P_z),P_z)$, $b'=\ra(b,P_x)$, $d'=\ra(d,P_y)$, $f'=\ra(f,P_z)$; the successor of a value $x$ in a set $X$, denoted $\ssucc(x,X)$,  is the smallest element in a set $X$ that is larger than or equal to $x$.

The following Lemma can be used for rank reduction on a set of poly-logarithmic size. 
\begin{lemma}\cite{GrossiORR09}
  \label{lemma:grossi}
  Suppose that we can access any element of an integer set $S$ in time $O(t_{acc})$ and $|S|=\log^{O(1)}n$. There is a data structure that  answers predecessor and successor queries on $S$ in $O(1+t_{acc})$ time and  uses $O(|S|\log\log n)$ additional bits.   
\end{lemma}

Suppose that we store the set $P$ of three-dimensional points such that $|P|=\log^{O(1)}n$ and every point of $P$ can be accessed in time $O(t_{acc})$. By Lemma~\ref{lemma:grossi}, we can answer successor queries on  $P_x$, $P_y$, and $P_z$ in $O(t_{acc})$ time using $O(\log\log n)$ bits per point.

\begin{lemma}
   \label{lemma:small0}
If a set $P$ contains $m=\log^{O(1)} n$ points in the rank space of $P$, then we can keep $P$ in a data structure  that uses $O(m\log\log n)$ bits and answers three-dimensional dominance range reporting queries and three-dimensional five-sided rage reporting queries in $O(k)$ time. This data structure  uses  a universal look-up table of size $o(n)$.
\end{lemma}
Lemma~\ref{lemma:small0} can be proved in exactly the same way as Lemma 7 in~\cite{N20arx}. Lemma 7 in~\cite{N20arx} 
is proved for dominance queries on $O(\log^2 n)$ points.
However exactly the same method  can be also used for five-sided queries and for any poly-logarithmic number of points. 

Combining Lemma~\ref{lemma:small0} and the observations after Lemma~\ref{lemma:grossi}, we have the following result.
\begin{lemma}
   \label{lemma:small0}
If a set $P$ contains $t=\log^{O(1)} n$ points and we can obtain the coordinates of any point in $P$ in time $O(t_{acc})$. There is a data structure  that uses $O(m\log\log n)$ bits additional bits and answers three-dimensional dominance range reporting queries and three-dimensional five-sided rage reporting queries in $O(k\cdot t_{acc})$ time. This data structure  uses  a universal look-up table of size $o(n)$.
\end{lemma}

\section{Additional Figures}
\label{sec:figures}

\begin{figure}
\centering
\begin{minipage}{.5\textwidth}
  \centering
  \includegraphics[width=.65\linewidth,page=1]{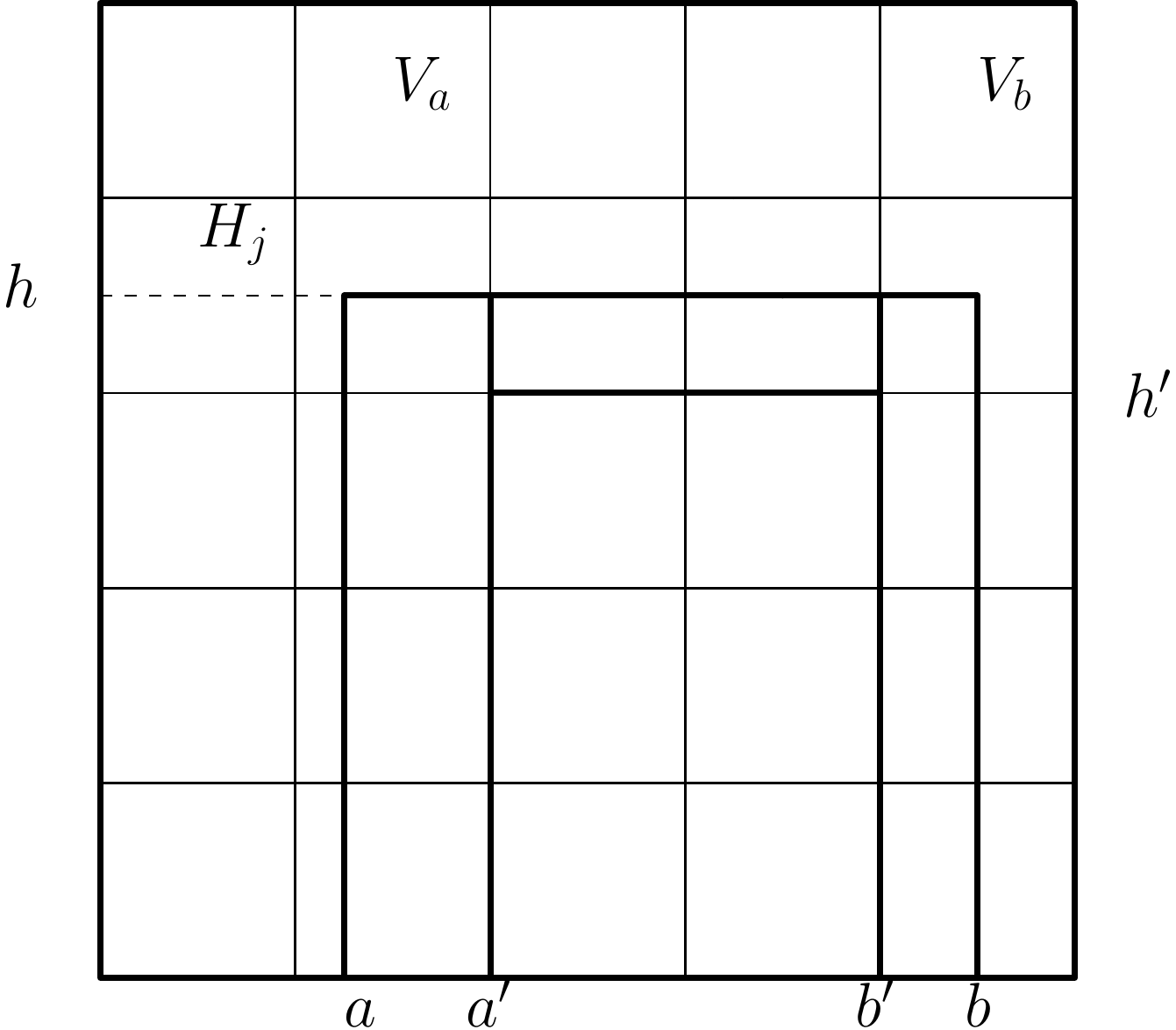}
  \label{fig:test1}
\end{minipage}%
\begin{minipage}{.5\textwidth}
  \centering
  \includegraphics[width=.65\linewidth,page=2]{grid}
  \label{fig:test2}
\end{minipage}
\caption{Left: decomposition of a four-sided query. Right: decomposition of a five-sided query. Only projections of points onto $(x,y)$-plane are shown.}
\label{fig:grid}
\end{figure}

\begin{figure}[tb]
  \centering
  \begin{minipage}{.5\textwidth}
    \centering
    \includegraphics[width=.85\linewidth,page=1]{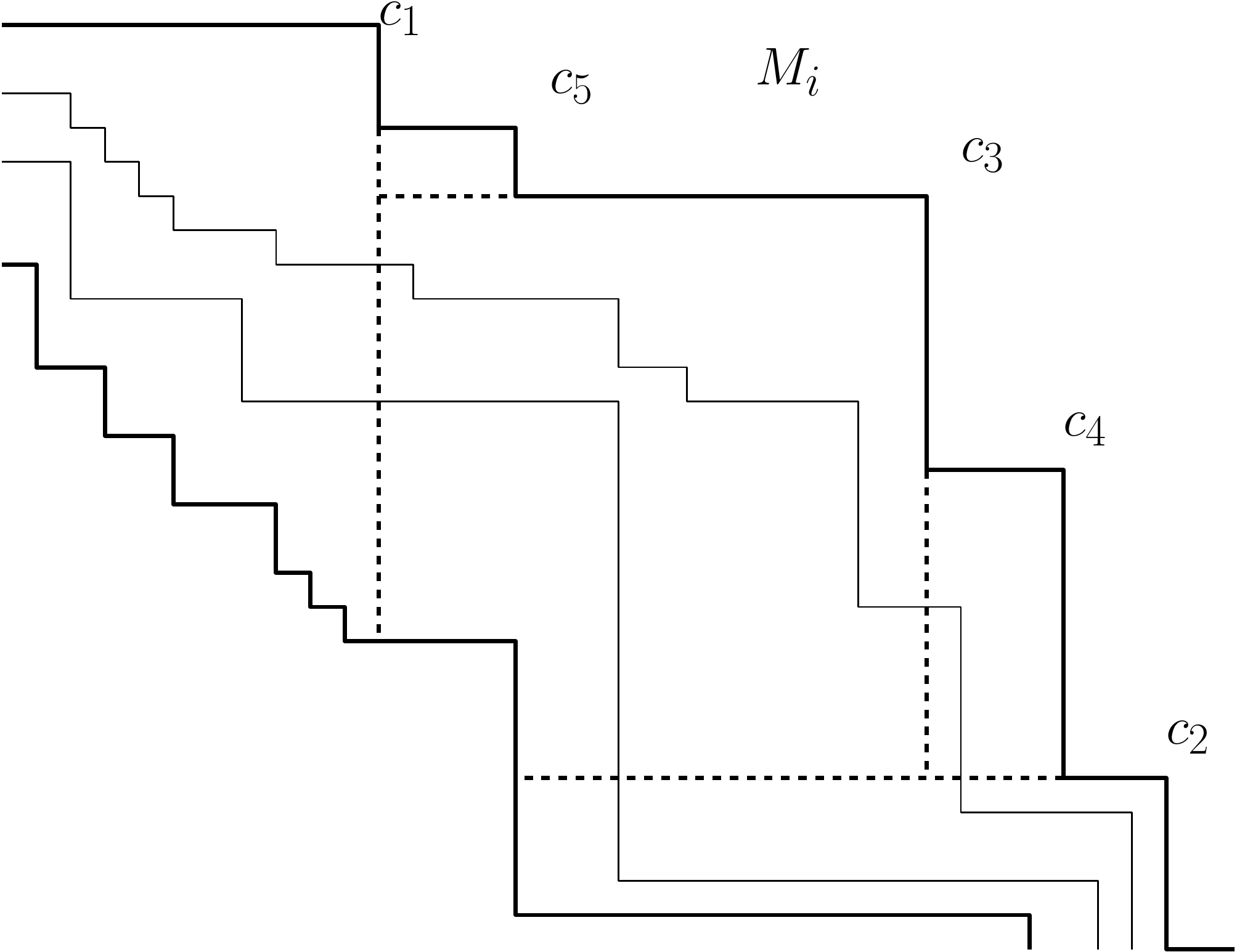}%
    $~$\\
    ({\bf a})
    \label{fig:cover1}
  \end{minipage}%
  \begin{minipage}{.5\textwidth}
    \centering
    \includegraphics[width=.85\linewidth,page=3]{shcut-cover}%
    $~$\\
    ({\bf b})
    \label{fig:cover2}
  \end{minipage}
  \caption{Division of a $t$-shallow cutting into regions. Left: shooting vertical and horizontal rays from corners of $M_i$. We assume that $d=3$ and $z(c_i)<z(c_{i+1})$ for $1\le i\le 5$. Right: regions associated to corners $c_3$, $c_4$, $c_5$ and partially $c_1$ are shown in different colors.The region of $c_1$ is empty. The region of $c_3$, shown in yellow, is divided into two rectangles. The regions of $c_5$ and $c_4$ consist of one rectangle each.}
  \label{fig:shcut-cover}
\end{figure}

\begin{figure}
\centering
\begin{minipage}{.5\textwidth}
  \centering
  \includegraphics[width=.65\linewidth,page=2]{range-tree}
  \label{fig:test1}
\end{minipage}%
\begin{minipage}{.5\textwidth}
  \centering
  \includegraphics[width=.45\linewidth,page=1]{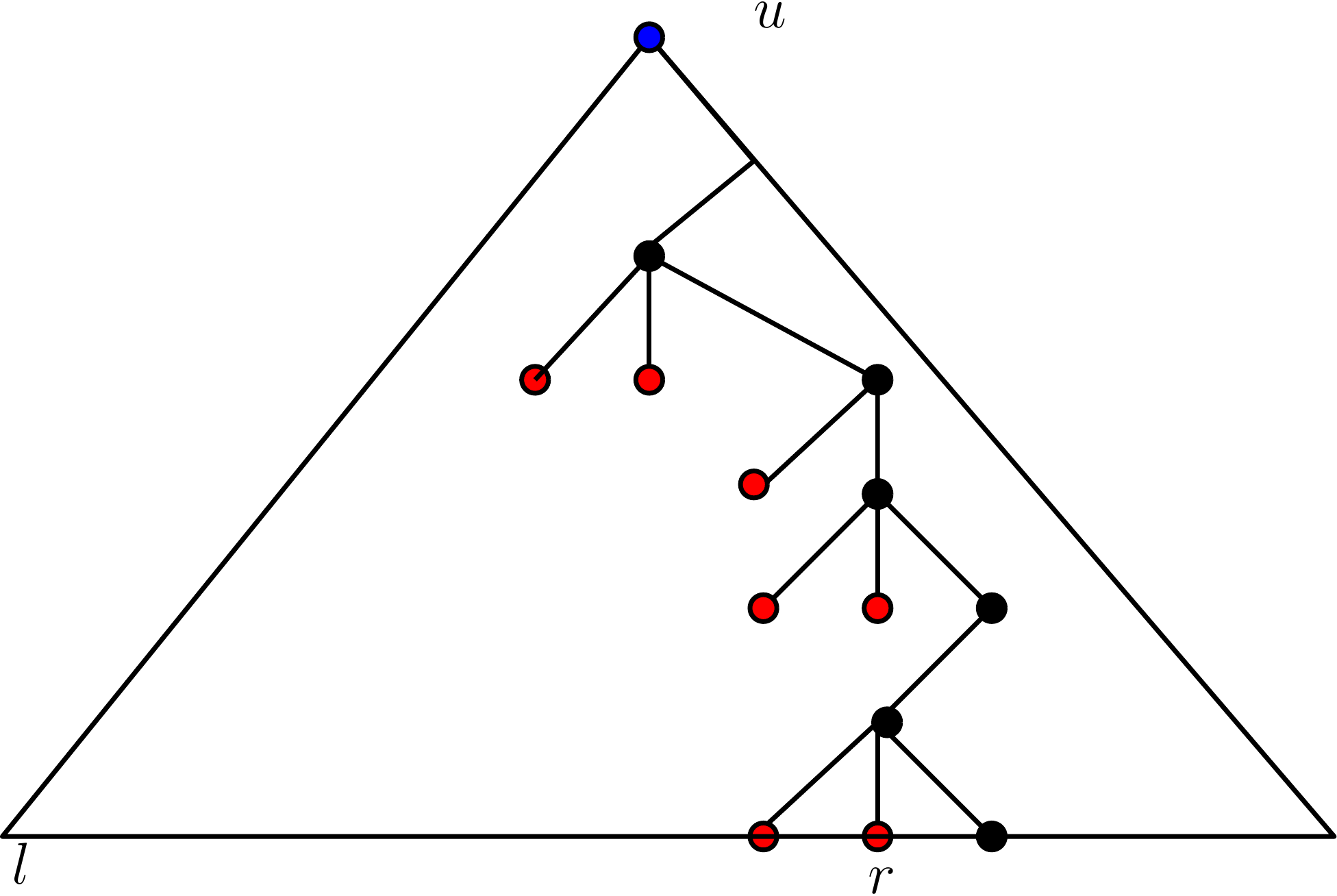}
  \label{fig:test2}
\end{minipage}
\caption{Left: Range tree and a path to $q_{z'}$. Triangles are subtrees corresponding to nodes of $T^0$. Red areas show the canonical decomposition of the $[1,q_{z'}]$. Right:  A subtree corresponding to a node $u\in T^0$ and its children.  Red nodes are a canonical decomposition of the node range $S(u,l,r)$.}
\label{fig:rangetree4d}
\end{figure}

\section{Proofs of Lemma~\ref{lemma:cov1} and Lemma~\ref{lemma:p2}}
\label{sec:coverlemmas}
\paragraph{Proof of Lemma~\ref{lemma:cov1}.}
\begin{proof}
  Consider a corner $c_j$ and a point $p\in list(C_j)$ that is not contained in any rectangle associated to a neighbor of $c_j$. The following cases are possible: (1) $p$ is contained in an empty region of some neighbor $c_l$ of $c_j$. By definition of a region, $l< j$, and $z(c_l)\ge z(c_j)$. Hence  $p \in list(C_l)$ and $p\in P_2$. (2) $p$  is dominated by (corners of ) at least $d$ non-empty regions. In this case $\pi(p)$ is dominated by at least $d$ corners $\pi(c_l)$ for some $l< j$. Since $z(c_l)>z(c_j)> z(p)$, $p$ is contained in the conflict list of every such $C_l$ and $p\in P_1$.  (3) $p$ is dominated by some corner of $M_{i+d}$. In this case we can show that  $p$ is in at least
  $d$ conflict lists:  for $l=1$, $2$, $\ldots$, $d$, the point $\pi(p)$ is dominated by some corner $\pi(c_{j_l})$ of $M_{i+l}$ such that $\pi(c_{j_l})$ is dominated by $\pi(c_j)$. If $\pi(c_{j_l})$ is dominated by $\pi(c_j)$, then $z(c_{j_l})> z(c_j)$. Since $p$ is in the conflict list of $c_j$, $z(p)<z(c_j)$ and $z(p)<z(c_{j_l})$. Hence $p$ is contained in the conflict lists of at least $d$ corners and $p\in P_1$. 
\end{proof}

\paragraph{Proof of Lemma~\ref{lemma:p2}}
\begin{proof}
  We assign $d$ dollars to every point in $list(C_i)$ for every cell $C_i$ of $\cC$.   The same point can appear in many lists, but the total number of elements in all conflict lists is $O(m)$ and our total budget is $O(m\cdot d)$ dollars.  We  assume that the cost of adding a point to $P_2$ is $d^2$ dollars, and we will show that $m\cdot d$ dollars are sufficient to construct $P_2$.  
  If the  region of $c_j$ is empty, then it contains at least $d^2$ corners $c_{j_l}$.  We charge $1$ dollar to every point in the conflict list of each $C_{j_l}$. Every corner $c_r$ is contained in  at most $d$ different regions:  By definition of a region,  a corner on $M_i$ can be contained only in the region of a corner on $M_j$ for $i> j\ge  i-d$. Regions of corners on the same level of maxima are disjoint. Thus every point is charged at most $d$ times. Hence $P_2$ contains $O(m/d)$ points.
\end{proof}

\section{Range Minima: Faster Queries in More Space}
\label{sec:minima2}
In this section we describe a data structure that uses $O(n\log\log n)$ words of space and answers queries in $O((k+1)\log\log n)$ time. We use the same recursive grid as in Section~\ref{sec:recur}.
Our approach is based on constructing a data structure for four-sided queries in every slab.

\begin{lemma}
  \label{lemma:compact3sid}
  There exists a data structure supporting capped four-sided queries in $O(\log\log n + k)$ time and $O(m\log m)$ bits of space where $m$ is the number of points in a slab. The data structure relies on a universal data structure for two-dimensional range selection queries. 
\end{lemma}
\begin{proof}
  W. l. o. g. we consider queries $[a,b]\times [0,c]\times [0,d]$. We construct a range tree with node degree $\log^{\eps} n$ on $x$-coordinates of points. We keep two  dominance data structures in every tree node. These data structures support queries $[0,b]\times [0,c]\times [0,d]$  and   $[a,+\infty)\times [0,c]\times [0,d]$. Additionally each node contains a data structure that stores modified points supports "narrow" four-sided queries of the form $[i_1\times i_2]\times [0,c]\times [0,d]$ where $1\le i_1\le i_2\le \log^{\eps}n$. For every point $p=(p_x,p_y,p_z)$ stored in a node $u$, the narrow queries data structure contains  a point $(i,p_y,p_z)$, such that $p$ is stored in the $i$-th child of $u$.  

  To answer a query we identify the leaves $l_a$ and $l_b$ holding the successor of $a$ and the predecessor of $b$ respectively. Let $w$ denote the lowest common ancestor of these two leaves. Let $w_i$ and $w_j$ denote the children of $w$ that are ancestors of $l_a$ and $l_b$.  We answer the query $[a,+\infty)\times  [0,c]\times [0,d]$ in $w_i$ and $[0,b]\times [0,c]\times [0,d]$ in $w_j$.  Additionally we answer a narrow four-sided  query $[i+1,j-1]\times [0,b]\times [0,d]$ on in the node $w$. The answer to the query contains all points from  $[a,b]\times [0,c]\times [0,d]$.

Dominance data structures are implemented as in Section~\ref{sec:slabdomin} and use $O(\log\log n)$ bits per point. We will show below that each narrow four-sided strucure also uses $O(\log \log n)$ bits per point. Since each point is stored twice on every level of the range tree and there are $O(\log m/\log\log n)$ levels, the total space usage is $O(m\log m)$ bits.   
\end{proof}

\begin{theorem}
  There exists  a data structure that supports five-sided three-dimensional range reporting queries in $O((k+1)\log\log n)$ time and uses $O(n\log\log n)$ space.\\
  The same data structure can be adjusted to support two-dimensional range minima queries in $O(\log \log n)$ time. 
\end{theorem}
\begin{proof}
  We use the recursive grid described in Section~\ref{sec:recur} and store the four-sided data structure for each slab.  A capped five-sided query can be represented as a union of at most four four-sided queries and a query to a top data structure. Hence a query is answered $O((k+1)\log\log n)$ time. If $k>\log n$, we use the data structure of Chazelle~\cite{Chazelle88} that uses $O(n \log \log n)$ words and supports queries in $O(\log n \log\log n + k\log \log n)$ time. If $k\ge  \log n$, the query time of Chazelle's structure can be simplified to $O(k \log \log n)$.

  All slab data structures on each recursion levels use $O(n \log n)$ bits in total. Since the depth of recursion is $O(\log\log n)$, the total space usage is $O(n\log\log n)$ words.  We store each slab data structure in the rank space of its slab. Each point in a slab can be "decoded" in $O(\log\log n)$ time. Hence  we can transform a query to the rank space of a slab in $O(\log \log n)$ time.
\end{proof}

It remains to describe how narrow four-sided queries are answered. The following property of two-dimensional $t$-shallow cuttings, very similar to Theorem~\ref{theor:shallowpart},  will be used in our method. 

We can prove the analogue of Theorem~\ref{theor:shallowpart} for 2-d points. To keep the description unified with the rest of this section, we consider points on the $(y,z)$-plane.
\begin{theorem}
  \label{theor:shallowpart2d}
Let $\cC$ denote a $t$-shallow cutting of a two-dimensional set $P$, $|P|=m$. There exists a subset $P'$ of $P$ and a set of 2-d rectangles $\cR=\{\,R_1,R_2,\ldots, R_s\,\}$, such that
\begin{description}
\item[(a)] 
  $|P'|\le m/d$
\item[(b)]
  Rectangles $R_i$ are unbounded along the $z$-axis.
\item[(c)]
  The conflict list of any cell, except for points from $P'$, is contained in $O(d)$ rectangles from $\cR$, \[list(C_i)\cap (P\setminus P')\subseteq (list(C_i)\cap R_{i_1})\cup (list(C_i)\cap R_{i_2})\cup\ldots\cup (list(C_i)\cap R_{i_g}\] for $g=O(d)$.
\item[(d)]
  Each rectangle contains $O(t\cdot d)$ points of $P$. 
\end{description}
\end{theorem}
\begin{proof}
  For a cell $C_i$ of the shallow cutting, let $c_i$ denote its upper right corner. We can assume w.l.o.g. that no cell $C_i$ is entirely contained in some other cell $C_j$. Hence no $c_i$ is dominated by $c_j$.  We will also assume that  all corners $c_i$ are sorted in increasing order by $y$-coordinates (and thus in decreasing order by $z$-coordinates). 

  We define $P_1$ to be the set of points stored in more than $d$ conflict lists. Since the total number of elements in all conflict lists is $O(n)$, $P_1$ contains $O(n/d)$ points. The set $\cR$ contains a rectangle $R_j=[c_{j-1}.y,c_j.y]\times (-\infty,+\infty)$ for every $c_j$ (we set $c_0.y=0$).  If a rectangle $R_j$ contains over $t\cdot d$ points, we add all points from $R'_j=[c_{j-1}.y,c_j.y]\times [0,c_j.z]$ to $P_2$ and remove $R_j$ from $\cR$. See Fig~\ref{fig:2d-shcut}.
  All points in  $R'_j$ are dominated by $c_j$; hence $R'_j$ contains at most $2t$ points. Hence there are $d/2$ points in $P$ for every point in $P_2$ and $|P_2|\le 2n/d$.  We set $P'=P_1\cup P_2$ so that $|P|=O(n/d)$.
  \begin{figure}
    \centering
    \includegraphics[width=.45\linewidth]{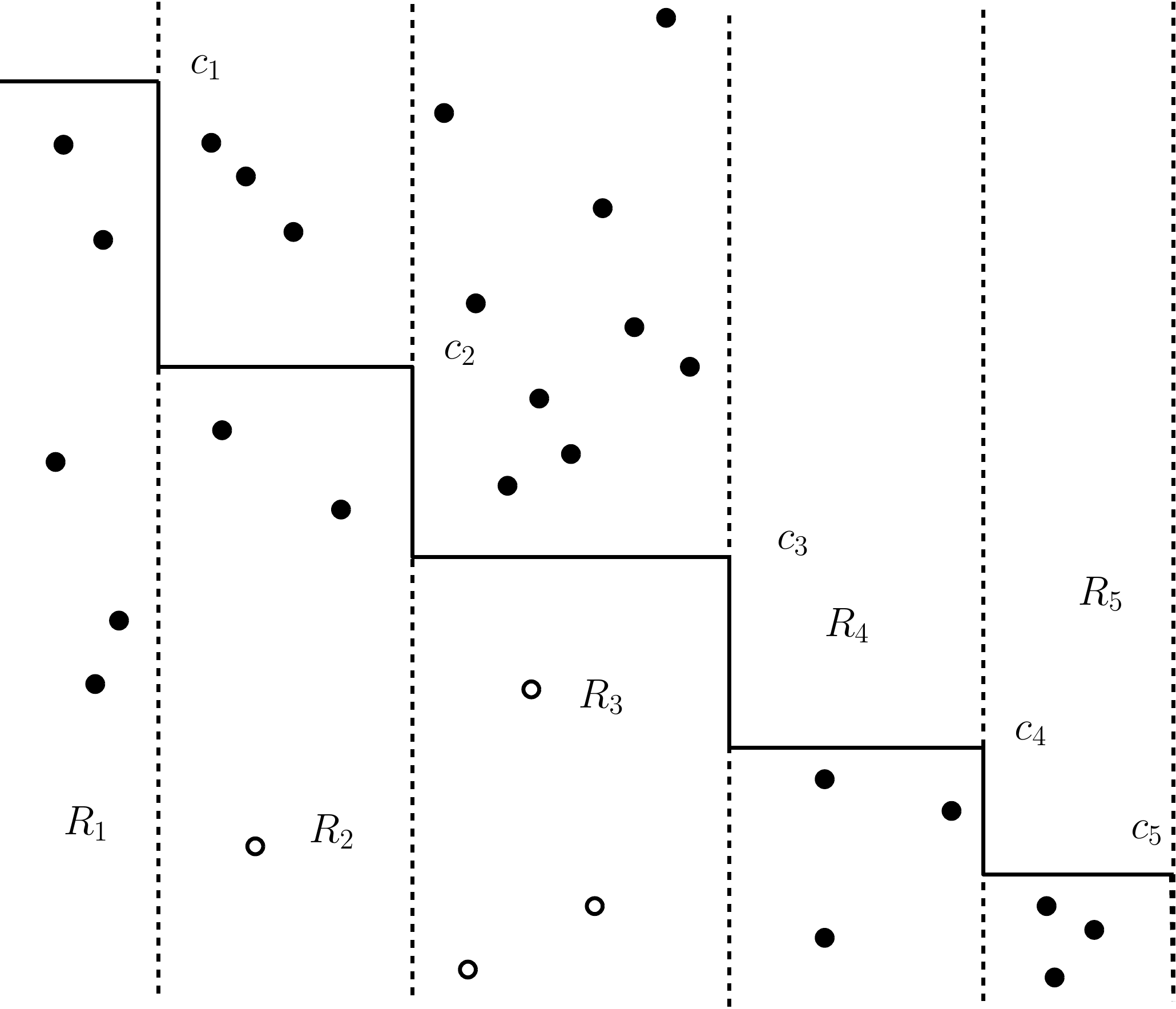}
    \caption{Example of a $t$-shallow cutting in two dimensions and its covering for $t=3$ and $d=3$. Points from $P'$ are depicted by hollow circles, all other points are depicted by filled circles. Dashed lines are boundaries of rectangles $R_i$. The hollow point in rectangle $R_2$ must be stored in conflict lists of $C_2$, $C_3$, and $C_4$. Therefore this point is in $P_1$. Hollow points in rectangle $R_3$ are in $P_2$ because $R_3$ contains over $t\cdot d$ points.}
    \label{fig:2d-shcut}
  \end{figure}
  Consider an arbitrary cell $C_i$ and points dominated by $c_i$.  Suppose that a point $p\in P$, dominated by $c_i$, is also dominated by $c_{i-l}$ for some $l\ge d$.  Since $z$-coordinates of corners increase when $i$ decreases, $p$ is dominated by all $c_j$, $i-l\le j\le i$. Hence $p$ is contained in at least $d$ cells and $p\in P_1$.
Consider  a point $p\in R_j$ for $i-d < j\le i$, such that $p$ is dominated by $c_i$. If $R_j$ contains at least $t\cdot d$ points, then $p\in P_2$.   
\end{proof}

\begin{lemma}
  \label{lemma:narrow4sided}
  There exists a data structure that answers narrow four-sided queries in $O(\log \log n+ k)$ time and uses
  $O(m\log\log n)$ bits of space, where $m$ is the number of points in the data structure. 
\end{lemma}
We combine the approach from~\cite{ChanNRT18} with a compact representation of two-dimensional shallow cuttings. 

Let $P_i$ denote the set of points $p$ such that $p.x=i$. We consider projections of points in $P_i$ onto the $(y,z)$ plane  and construct a  2-d $t$-shallow cutting for $t=\log^3 n$.   Using Theorem~\ref{theor:shallowpart2d}, we construct a subset $P_i'\subset P_i$  such that $P'_i=O(|P_i|/d)$ for $d=\log^2 n$.
For each cell $C_j$, $list(C_j)\setminus P'_i$ is contained in $O(d)$ rectangles unbounded in $z$-direction.

All points from $P'_i$ are stored in the data structure from~\cite{ChanLP11} that uses $O(|P'|\log^{1+\eps} n)$ bits and supports queries in $O(\log \log n + k)$ time.  We follow the method of ~\cite{ChanNRT18} to support four-sided narrow queries on $\cup(P_i\setminus P'_i)$. The only difference with~\cite{ChanNRT18} is that points in every group $G$ (defined as in of~\cite{ChanNRT18}) are stored in the rank space. Additionally for every point $p$ in $G$ we store: (a) the index of $i$ of $P'_i$, such that $p\in P'_i$, (b) the identifier of the rectangle $R\in \cR$ that contains $p$,  (c) the $y$-rank of $p$ in $R$, and (d) the cell $C_f$, such that $p\in C_f$.

Every group $G$ contains points from $O(\log^{2\eps}n)$ different cells. Hence we can specify $C_f$ using $O(\log \log n)$ bits. Since
$rlist(C_f)$ consists of $O(d^2)$ rectangles by Theorem~\ref{theor:shallowpart2d}, we can specify $R$ using $O(\log d)=O(\log\log n)$ bits. We can store the index of $P_i$ n $O(\log\log n)$ bits because a node has $\log^{\eps}n$ children. And we can store the $y$-rank of $p$ in $R$ using $O(\log\log n)$ bits because the number of elements in $R$ is poly-logarithmic. Hence we spend $O(\log\log n)$ bits per point. Additionally we store the list of rectangles $rlist(C)$ for every cell $C$. All lists use $O(m\cdot d\log^3 n/t)=O(m)$ bits.  

We also need an $O(n\log\log n)$-word universal data structure for capped selection, implemented as in Lemma~\ref{lemma:smallsel}, part (b). When this data structure and the above information are available, we can retrieve the coordinates of a point in $O(\log\log n)$ time.
Thus, as explained in Section~\ref{sec:rankspace}, we can support queries on the rank-reduced points of  $G$ in $O(\log\log n)$ time per reported point. A general narrow four-sided query can be reduced to a query on a group $G$~\cite{ChanNRT18}. Hence, we can answer narrow four-sided queries in $O((k+1)\log\log n)$ time.

\begin{lemma}
  There exists a data structure that uses $O(n\log\log n)$ space and supports five-sided three-dimensional queries in $O((k+1)\log\log n)$ time, where $k$ is the number of reported points. The same data structure can be modified to support two-dimensional range minima queries in $O(\log\log n)$ time. 
\end{lemma}
\begin{proof}
  We use recursive grid defined in Section~\ref{sec:recur} and store the data structure for four-sided queries of Lemma~\ref{lemma:compact3sid} in every slab. Points every slab are reduced to rank space.  The space usage of the data structure in bits is $S(n)=O(n\log n)+ 2\sqrt{n/\log^2 n}S(\sqrt{n\log^2 n})$. Hence $S(n)=O(n\log\log n)$, see e.g.,\cite{AlstrupBR00,Nekrich07algorithmica}.. 

  A five-sided query can be reduced to a query in a single slab or to at most four four-sided queries and one query to the top data structure. The query time  is the same as in Lemma~\ref{lemma:compact3sid}.

 As explained in Section~\ref{sec:slabdomin}, we can adjust our result to support range minima queries in time $O(\log\log n)$.
\end{proof}

\section{Data Structure for Capped Selection Queries}
\label{sec:smallsel}
\begin{lemma}
  \label{lemma:smallsel}
  Let $P$ be a set of $n$ two-dimensional points and $\cR$ be a set of $m_R$ rectangles. There exists a data structure that uses $O(s(n)+ m_R\log^{1+\eps} n)$  words and answers capped range  selection queries $(R,f)$ where $R\in \cR$  and $f=O(\log^{10} n)$, in time $q(n)$.  The following trade-offs between  $q(n)$ and $s(n)$ are possible:\\
  (a) $s(n)=O(n)$ and $q(n)=O(\log^{\eps}n)$\\
  (b) $s(n)=O(n\log\log n)$ and $q(n)=O(\log\log n)$\\
  (c) $s(n)=O(n\log^{\eps}n)$ and $q(n)=O(1)$\\
\end{lemma}

We store a compact range tree with node degree $\log^{\eps} n$ for some constant $\eps>0$  on a set $P$.  Recall that a standard range tree is a balanced tree on the $x$-coordinates of points. Every internal node stores the set of points $S(u)$ that contains all points whose $x$-coordinates are in the leaf descendants of $u$. Although a compact range tree does not store $S(u)$ in explicit form, it supports operations $\noderange(y_1,y_2,u)$ and $\point(u,i)$. The former operation identifies the range $[h_1..h_2]$  such that all points stored in $S(u)$ at positions $h_1$, $h_1+1$, $\ldots$, $h_2$ have $y$-coordinates in the interval $[y_1,y_2]$. The operation $point(u,i)$ returns the coordinates of the $i$-th point in $S(u)$ (assuming that points in $S(u)$ are sorted by their $y$-coordinates). 
Different trade-offs between the space usage $s_{\mathtt{comp}}(n)$ of the compact tree and the cost $t_{\mathtt{comp}}(n)$ of $\point(u,i)$ are possible: either (i) $s_{\mathtt{comp}}(n)=O(n)$ and $t_{\mathtt{comp}}(n)=O(\log^{\eps}n)$ or (ii) $s_{\mathtt{comp}}(n)=O(n\log\log n)$ and $t_{\mathtt{comp}}(n)=O(\log\log n)$ or (iii) $s_{\mathtt{comp}}(n)=O(n\log^{\eps} n)$ and $t_{\mathtt{comp}}(n)=O(1)$.

For every node in the range tree we store a data structure supporting range $\tau$-selection queries:  for any
$y$-range $[y_1,y_2]$ and  any $f\le \tau$,  we can return the index of the point with the $f$-th smallest $x$-coordinate in $S(u)[y_1..y_2]$. When $\tau=\log^{O(1)}n$, we can support range $\tau$-selection queries in $O(1)$ time using $O(\log\log n)$ bits per point~\cite{NavarroRS14,GawrychowskiN15}. 

Every covering rectangle $R=[a,b]\times [c,d]$ is divided into $O(\log^{1+\eps} n)$ smaller rectangles: we can represent $[a,b]$ as a union of $O(\log^{1+\eps}n)$ intervals $[a_l,b_l]$ where $a_l$ and $b_l$ are the leftmost and the rightmost leaf descendants of   some node $u_l$ in the range tree. Let $[c_i,d_i]=\noderange(c,d,u_i)$.  We store the coordinates of each $R_i=[a_i,b_i]\times [c_i,d_i]$,  and the number of points  $f_i=|P\cap R_i|$ for each $R_i$. We also compute the prefix sums $F_i=\sum_{j=1}^i f_i$ for all $i$.

A query $(R,f)$ is answered as follows.  We consider the decomposition of $R$ into rectangles $R_i$ and find the index $j$, such that $F_{j-1}< f \le F_j$. Let $f'=f-F_{j-1}$. Using the range selection data structure, we can find the index of the $f'$-th leftmost point in $S(u)[c'..d']$ where $[c',d']=\noderange(c,d)$.  Then we can obtain the point by answering the $\point$ query.

The query time is dominated by the operations on the range tree.  Hence we obtain the same space-time trade-offs for the capped range selection as for the compact range tree.


\no{
We construct a binary  range tree on $x$-coordinates of points.   For every node $u$ of the range tree we store the compact data structure supporting range selection queries. For every
$y$-range $[y_1,y_2]$ and we can return the index of the point with the $f$-th smallest $x$-coordinate for $f\le \log^7 n$. We also store a symmetric data structure that returns the index of the point with the $f$-th largest $x$-coordinate for every $y$-range $[y_1,y_2]$.  Both data structures are implemented using the result from ~\cite{???}, so that queries are supported in $O(1)$ time and the space usage is $O(\log\log n)$ bits per point.   

A query $([a,b]\times [c,d],\ell)$ is answered as follows: let $w$ denote the lowest common ancestor of the leaves holding $a$ and $b$. We find the node range $[c_r,d_r]$ of $[c,d]$ in the right child $u_r$ of $v$. Since all points have different $y$-coordinates, the number of
point with different $y$-coordinates in the $y$-range  of $u_r$ is $\ell_r=d_r-c_r+1$.
If $\ell_r\ge \ell$, we find the index $i$ of the point with the $\ell_r-\ell$-th smallest $x$-coordinate in $[c_r,d_r]$. Otherwise we find the point with the 
}

\section{Analysis of Four-Dimensional Range Reporting}
\label{sec:analysis4d}
We need to prove that $S=\sum_{i=0}^h \frac{1}{\log\alpha_i} =O(1)$.  We define the sequence $f(i)$ as follows: $f(0)=0$, $f(i)=\min\{\,x\,| \alpha_x\le \log(\alpha_{f(i-1)})\,\}$.  Let $\sigma_i=\sum_{j=f(i)}^{f(i+1)-1} \frac{1}{\log \alpha_i}$.  The sum $\sigma_i$ has $O(\log\log \alpha_{f(i)})$ terms. By definition of $f()$, $\alpha_j> \log(\alpha_{f(i)})$ for $f(i)\le j< f(i+1)$. Hence each term in $\sigma_i$ is smaller than $\frac{1}{(\log\log \alpha_{f(i)})^2}$  and $\sigma_i=O(\frac{1}{\log\log \alpha_{f(i)}})$.

We can represent $S$ as the sum of $\sigma_j$. By the above analysis $S=O(\sum_{f(i)\le h}\frac{1}{\log\log \alpha_{f(i)}})$. Let $l$ denote the number of terms in the latter sum. 
Let $\beta_i=\alpha_{f(i)}$. Then $S=O(S')$  where $S'=\sum_{i=0}^{l-1} \frac{1}{\log\log \beta_{l-i}}$. By definition of $f(i)$ $\beta_i\le \log \beta_{i-1}$  and $\log \log \beta_i< (1/2) \log\log \beta_{i-1}$.  Hence $\frac{1}{\log\log \beta_i}> \frac{2}{\log\log \beta_{i-1}}$.     Since $\beta_l=O(1)$, $\frac{1}{\log\log \beta_l}=O(1)$. Hence the sum $S'$ can be bounded by a decreasing geometric sequence with constant first term. Therefore $S=O(S')=O(1)$.

\section{Space-Efficient Four-Dimensional Range Reporting}
\label{sec:space4d}
In this section we describe a data structure with $O(n\log^{2+\eps} n)$ space and answers 4d orthogonal range reporting queries in optimal time.

We will say that a point $p$ is on a 4d-narrow grid if the fourth coordinate of $p$ is bounded by $\mu=\gamma^{(\log\log n)^2}$ where $\gamma=\log^{\eps}n$. A query $Q=[x_1,x_2]\times [y_1,y_2]\times [0,z]\times [z_1',z_2']$ is called a $(2,2,1,2)$-sided query or a $7$-sided query. The projection of $7$-sided $Q$ on $z$-axis is a half-open interval, the projections of $Q$ on $x$, $y$-, and $z'$-axes are closed intervals. We will show that $7$-sided queries can be supported in optimal time using $O(n\log^{1+\eps}n)$ space:
First we show that dominance queries on a 4d-narrow grid can be supported in $O((\log\log n)^2)$ time using $O(n\log^{\eps}n)$ space\footnote{To simplify the notation we sometimes ignore the time needed to report points in this section.  Whenever we say that reporting queries are supported in time $O(f(n))$, we impy the reporting time $O(f(n)+k)$.}. 
Then we apply a lopsided grid approach~\cite{ChanLP11} and obtain a data structure that supports $(2,2,1,2)$-sided queries on a narrow 4d-grid in $O((\log\log n)^2)$ time and uses $O(n\log^{\eps}n$ space.
Using range trees on the fourth coordinate with node degree $\mu$, we extend this result to a data structure that answers  $(2,2,1,2)$-sided queries in $O(\log n/\log\log n)$ time and $O(n\log^{1+\eps} n)$ space. Finally we obtain the  result for general four-dimensional reporting queries with optimal time and $O(n\log^{2+\eps})$ space.

\begin{lemma}
  \label{lemma:small4d-domin}
  For any $m\le n$ there exists a data structure that uses $O(m\log^{\eps} m)$ words and supports dominance queries on a 4d-narrow grid in time $O((\log\log n)^2)$, where $m$ is the number of points in the data structure. 
\end{lemma}
\begin{proof}
  We use the same method as in Theorem~\ref{theor:opttime}, but we need to adjust some parameters for the case when $m$ is small.  Recall that we can directly apply the strategy of Theorem~\ref{theor:opttime} only in the case when $m\ge \gamma^{2\rho_0}$: otherwise it is not possible to store even $\Theta(1)$ words for all possible node ranges of $0$-nodes.
Fortunately we can first reduce all points to the rank space and then answer queries in $0$-nodes in $O(\log\log m)$ time per node. Hence  the node degree of $T_0$ can be decreased. A more detailed description is below. 

  We set $\alpha_0=\log\log m$ and $\ell=\log\log m$. As before, $\alpha_i=\sqrt{\alpha_{i-1}}\log^2 \alpha_{i-1}$, $\rho_i=\gamma^{\alpha_i}$ and $t_i=\rho_i^4$.  The range tree on the fourth coordinate has node degree $\gamma=\log^{\eps} n$ and we store the same data structures as in Section~\ref{sec:fourdim} in the nodes of $T$.  Points of the input set are stored in the rank space. Hence we can support orthogonal point location queries in $0$-nodes in time $O(\log \log m)$.  The search procedure is the same as in Section~\ref{sec:fourdim}. In order to answer a query, we need to locate $q'$ in $O(\log_{\gamma} \mu)$ shallow cuttings $\cC(u,l,r)$. For all nodes $u$, such that the height of $u$ does not exceed $\ell$
  we answer a point location query in $O(\log\log m)$ time per node.  For all other nodes, we spend $O(1)$ time per node.
  Hence we can find the cells $C_u$ of $\cC(u,l,r)$ for all $u$ in time $O(\log_{\gamma}\mu +(\log\log m)^2)=O(\log_{\gamma}\mu)$.
  When all $C_u$ are known we can answer the three-dimensional dominance query in $O(1)$ time per point.    To transform the query to the rank space, we need to answer $O(1)$ successor queries. This takes additional $O(\log\log n)$ time. Hence the total query time is $O(\log_{\gamma}\mu+\log\log n)=O((\log\log n)^2)$. The space usage  can be analyzed in the same way as in Section~\ref{sec:fourdim}. 
\end{proof}

\begin{lemma}
  \label{lemma:small4d-7side}
  For any $m\le n$ there exists a data structure that uses $O(m\log^{4+\eps}m)$ words and supports $(2,2,1,2)$-sided queries on a 4d-narrow grid in time $O((\log\log n)^2)$, where $m$ is the number of points in the data structure. 
\end{lemma}
\begin{proof}
 Using standard techniques, we can extend  a data structure that answers queries of the form $[0,b]\times [0,c]\times [0,d]\times [0,e]$ to a data structure that answers queries of the form $[a,b]\times [0,c]\times [0,d]\times [0,e]$. The transformation does not increase the query time and increases the space usage by $O(\log m)$ factor. We can use this technique for any coordinate. Applying this transformation four times to Lemma~\ref{lemma:small4d-domin}, we obtain the result of this lemma. 
\end{proof}

\begin{lemma}
  \label{lemma:4d6side}
  For any $m\le n$ there exists a data structure that uses $O(m\log^{\eps}n)$ words and supports $(2,1,1,2)$-sided queries on a 4d-narrow grid in time $O((\log\log n)^2)$, where $m$ is the number of points in the data structure. 
\end{lemma}
\begin{proof}
  Let $A=2^{\log^{1-\eps}n}$ and $\tau=\mu^5$. We store a tree with node degree $O(A)$ on $x$-coordinates of points. Every tree leaf contains $O(A)$ points. Let $S(u)$ denote the set of points stored in leaf descendants of a node $u$.  We divide each $S(u)$ into columns and rows.  A point $p\in S(u)$ is assigned to column $V_i$ if it is stored in the $i$-th child of $u$.  We also divide $S(u)$ into $n/(A\cdot \tau)$ rows of size $O(A\cdot \tau)$.  An intersection of a row and a column is called a grid cell. For every range $[z'_1,z'_2]$ of $z'$-coordinates, we keep a top data structure $D^t[z_1,z_2]$ organized as follows. For every cell $G$, we consider all points $p$ in $G$ such that $p.z'\in [z'_1,z'_2]$,  and select $\log n$ points  with the smallest $z$-coordinates. All selected points are stored in the data structure $D^t[z'_1,z'_2]$. $D^t[z'_1,z'_2]$ supports three-dimensional five-sided range reporting queries:  given a query $[a,b]\times [c,d]\times [0,e]$, we can report all $p\in D^t[z'_1,z'_2]$  satisfying $p.x\in [a,b]$, $p.y\in [c,d]$, and  $p.z\le e$. Each $D^t[,]$ contains $O((n/\tau)\log n)$ points. Using the result from~\cite{ChanLP11}, each $D^t[\cdot,\cdot]$ can be implemented in $O((n/\tau)\log^{1+\eps} n)$ space so that queries are answered in $O(\log\log n)$ time.

  Each row that contains at least $\tau$ points and every leaf node that contains at least $\tau$ points, is recursively divided in the same way. If a row or a leaf node contains at most $\tau$ points, we keep all its points in the data structure of Lemma~\ref{lemma:small4d-7side}.

  To answer a query $Q=[a,b]\times [0,d]\times [0,e]\times [f,g]$, we find the lowest common ancestor of the leaves that contain $a$ and $b$ respectively. If $u$ is an internal node, we answer a query on $S(u)$. Since $u$ is the lowest common ancestor of $a$ and $b$, $a$ ad $b$  are stored in different columns $C_a$ and $C_b$ of $S(u)$. If the query is entirely contained in one row $R_0$, we answer the query using the recursive data structure for $R_0$. Otherwise we divide the query into at most four parts. We answer four-dimensional dominance queries $[a,+\infty)\times [0,d]\times [0,e]\times [f,g]$ and $[0,b]\times [0,d]\times [0,e]\times [f,g]$ in columns $C_a$ and $C_b$ respectively. Let $l$ be the largest index, such that $Q$ overlaps with the row $R_l$. We answer a query $[a,b]\times [0,d]\times [0,e]\times [f,g]$ using the recursive data structure for the row $R_l$.  Finally we answer the central query $Q'=Q\setminus (R_l\cup C_a\cup C_b)$. $Q'$ is the part of the query range that is not in included into the row $R_l$ or the columns $C_a$ and $C_b$.  This query can be answered using the top data structure $D^t[f,g]$. 
  
  The total query time satisfies the recursion $q(n)=q(2^{\log^{1-\eps}n})+ O(\log\log n)$. Since the recursion depth is a constant and the query time in the base case is $O(\log\log n)$, $q(n)=O(\log\log n)$. The space usage in bits satisfies the recursion $s(n)=O(n\log^{1+\eps}n)+ \log^{\eps}n\cdot s(2^{1-\log^{\eps}n})$. Let $r(n)=s(n)/n$, then $r(n)=O(\log^{1+\eps}n) + r(2^{1-\log^{\eps} n})$. In the base case $r(n)=O(\log^{1+4\eps}n)$. Since the recursion depth is constant, $r(n)=O(\log^{1+4\eps}n)$ Hence the space usage is $O(n\log^{4\eps}n)$ words.  If we replace $\eps$ with $\eps'=\eps/4$ in the above proof, we obtain the desired result. 
\end{proof}

\begin{lemma}
  \label{lemma:4d7side}
  For any $m\le n$ there exists a data structure that uses $O(m\log^{\eps}n)$ words and supports $(2,2,1,2)$-sided queries on a 4d-narrow grid in time $O((\log\log n)^2)$, where $m$ is the number of points in the data structure. 
\end{lemma}
\begin{proof}
  We can use the same lopsided grid and the same method as in ~\ref{lemma:4d6side}.
\end{proof}

\begin{theorem}
  \label{theor:4d-space}
  There exists an $O(n\log^{1+\eps} n)$ space data structure that answers $(2,2,1,2)$-sided queries in $O(\log n/\log\log n+k)$ time.\\
  There exists an $O(n\log^{2+\eps} n)$ space data structure that answers four-dimensional orthogonal range reporting  queries in $O(\log n/\log\log n+k)$ time.
\end{theorem}
\begin{proof}
  To prove the first statement, we construct a range tree $T_{\mu}$  with node degree $\mu$ on the fourth coordinate.  We keep the data structure of Lemma~\ref{lemma:4d7side} that supports $7$-sided queries on the narrow grid in every node of $T_{\mu}$. This data structure keeps all points stored in the node $u$ with the following change: the fourth coordinate of every point is replaced by the index of the child where $p$ is stored; that is, for each $p\in S(u)$ we replace $p.z'$ with $p.in$ such that $p\in S(u_{p.in})$.   Now any $7$-sided query can be answered by  answering $O(\log_{\mu} n)=O(\log n/(\log\log n)^3)$ queries to node data structures. Given a query $[a,b]\times [c,d]\times [0,e]\times [f,g]$ we can represent $[f,g]$ as a union  of $O(\log_{\mu}n)$ node ranges $S(u,i,j)$ where $u$ is a node on the path from $f$ to the lowest common ancestor of $f$ and $g$ (resp. on the path from $g$ to the lowest common ancestor of $f$ and $g$). We can find all $p\in S(u,i,j)$ such that $a\le p.x \le b$, $c\le p.y \le d$ and $p.z\le e$ using the data structure for narrow $7$-sided queries on $S(u)$. Each one of $O(\log_{\mu} n)$ queries takes $O((\log\log n)^2)$ time by Lemma~\ref{lemma:4d7side}. Hence the total query time is $O(\log n/\log\log n)$. Every point is stored in $O(\log_{\mu} n)$ nodes; the data structure in each internal node $u$ uses $O(\log^{\eps}n)$ words per point. Hence $T_{\mu}$ with all additional structures uses $O(n\log^{1+\eps} n)$ words of space.  
  
  The result for $(2,2,1,2)$-sided queries can be extended to the data structure supporting general $(2,2,2,2)$-sided queries using the standard range tree. The query time remains unchanged and the space usage increases by $O(\log n)$ factor.
  Hence, we can support four-dimensional orthogonal range reporting queries in $O(\log n/\log \log n)$ time and $O(n\log^{2+\eps} n)$ space. 
\end{proof}

We can also generalize our result to $d$-dimensional orthogonal range reporting queries for any $d\ge 4$.
\begin{theorem}
  \label{theor:dspace}
    For any $d\ge 4$ there exists an $O(n\log^{d-2+\eps} n)$ space data structure that answers $d$-dimensional orthogonal range reporting  queries in $O((\log n/\log\log)^{d-3} n+k)$ time.\\
\end{theorem}

\bibliographystyle{plain}
\bibliography{range-search}
\end{document}